\documentclass[11pt]{article}

\usepackage{etoolbox}
  
\usepackage[letterpaper,margin=1in]{geometry}

\usepackage{bbold}
\usepackage{amsmath}
\usepackage{amsthm}
\usepackage{amssymb}
\usepackage{pgffor}
\usepackage[inline]{enumitem}
\usepackage{mathtools}
\usepackage{bm}
\usepackage{microtype}
\usepackage{xspace}
\usepackage{xfrac}

\usepackage{times}

\usepackage{array}

\usepackage{url}
\usepackage{hyperref}

\hypersetup{colorlinks, linkcolor=darkblue, citecolor=darkgreen, urlcolor=darkblue}

\usepackage{float}
\usepackage{graphicx}

\graphicspath{{graphics/}}

\usepackage{tikz}

\usepackage[framemethod=tikz]{mdframed}

\usetikzlibrary{calc}
\usetikzlibrary{decorations.pathreplacing}
\usetikzlibrary{decorations.pathmorphing}
\usetikzlibrary{shapes.geometric}
\usetikzlibrary{fadings,decorations,fit}

\usepackage[vlined,ruled,algo2e]{algorithm2e}

\usepackage[textsize=scriptsize]{todonotes}

\definecolor{darkblue}{rgb}{0,0,0.38}
\definecolor{darkred}{rgb}{0.6,0,0}
\definecolor{darkgreen}{rgb}{0.1,0.35,0}

\DeclareMathOperator{\argmax}{argmax}
\DeclareMathOperator{\argmin}{argmin}

\makeatletter
\newcommand{\labeltarget}[1]{\Hy@raisedlink{\hypertarget{#1}{}}}
\makeatother

\def\OPT{\mathsf{OPT}}
\def\CG{\text{CG}}

\newtheorem{theorem}{Theorem}
\newtheorem{lemma}[theorem]{Lemma}

\newtheorem{proposition}[theorem]{Proposition}
\newtheorem{definition}[theorem]{Definition}

\newtheorem{corollary}[theorem]{Corollary}

\newtheorem{observation}[theorem]{Observation}

\newcommand{\Lcross}{\ensuremath{L_{\text{cross}}}}
\newcommand{\Lcr}{\ensuremath{L^{\text{crit}}_{\text{cross}}}}
\newcommand{\Lncr}{\ensuremath{L^{\text{no-crit}}_{\text{cross}}}}

\newcommand{\Lin}{\ensuremath{L_{\text{in}}}}
\newcommand{\Lup}{\ensuremath{L_{\text{up}}}}
\newcommand{\Vcr}{\ensuremath{V_{\text{crit}}}}

\newcommand{\across}{\ensuremath{\alpha_{\text{cross}}}}
\newcommand{\acr}{\ensuremath{\alpha^{\text{crit}}_{\text{cross}}}}
\newcommand{\ancr}{\ensuremath{\alpha^{\text{no-crit}}_{\text{cross}}}}
\newcommand{\ain}{\ensuremath{\alpha_{\text{in}}}}
\newcommand{\aup}{\ensuremath{\alpha_{\text{up}}}}

\newcommand{\Ltwol}{\ensuremath{L_{2l}}}
\newcommand{\Lonel}{\ensuremath{L_{0/1l}}}
\newcommand{\atwol}{\ensuremath{\alpha_{2l}}}
\newcommand{\aonel}{\ensuremath{\alpha_{0/1l}}}

\DeclareMathOperator{\E}{\mathbb{E}}
\DeclareMathOperator{\cov}{cov}
\DeclareMathOperator{\supp}{supp}

\title{Improved Approximation for Tree Augmentation:\\ Saving by Rewiring}

\author{
Fabrizio Grandoni\thanks{
IDSIA, Lugano, Switzerland.
Email: \href{mailto:fabrizio@idsia.ch}
{fabrizio@idsia.ch}. Partially supported by Swiss National Science Foundation grant 200021\_159697.
}
\and 
Christos Kalaitzis\thanks{
Department of Mathematics, ETH Zurich, Zurich, Switzerland.
Email: \href{mailto:chistos.kalaitzis@ifor.math.ethz.ch}
{christos.kalaitzis@ifor.math.ethz.ch}.
Supported by Swiss National Science Foundation grant 200021\_165866.
}
\and
Rico Zenklusen\thanks{
Department of Mathematics, ETH Zurich, Zurich, Switzerland.
Email: \href{mailto:ricoz@math.ethz.ch}
{ricoz@math.ethz.ch}.
Supported by the Swiss National Science Foundation grant
200021\_165866.
}
}
 
\begin{document}

\maketitle

\begin{abstract}
The Tree Augmentation Problem (TAP) is a fundamental network design problem in which we are given a tree and a set of additional edges, also called \emph{links}. The task is to find a set of links, of minimum size, whose addition to the tree leads to a $2$-edge-connected graph. A long line of results on TAP culminated in the previously best known approximation guarantee of $1.5$ achieved by a combinatorial approach due to Kortsarz and Nutov [ACM Transactions on Algorithms 2016], and also by an SDP-based approach by Cheriyan and Gao [Algorithmica 2017]. Moreover, an elegant LP-based $(1.5+\epsilon)$-approximation has also been found very recently by Fiorini, Gro\ss, K\"onemann, and Sanit\'a [SODA 2018]. In this paper, we show that an approximation factor below $1.5$ can be achieved, by presenting a $1.458$-approximation that is based on several new techniques. 

By extending prior results of Adjiashvili [SODA 2017], we first present a black-box reduction to a very structured type of instance, which played a crucial role in recent development on the problem, and which we call $k$-wide. Our main contribution is a new approximation algorithm for $O(1)$-wide tree instances with approximation guarantee strictly below $1.458$, based on one of their fundamental properties: wide trees naturally decompose into smaller subtrees with a constant number of leaves. Previous approaches in similar settings rounded each subtree independently and simply combined the obtained solutions. We show that additionally, when starting with a well-chosen LP, the combined solution can be improved through a new ``rewiring'' technique, showing that one can replace some pairs of used links by a single link. We can rephrase the rewiring problem as a stochastic version of a matching problem, which may be of independent interest. By showing that large matchings can be obtained in this problem, we obtain that a significant number of rewirings are possible, thus leading to an approximation factor below $1.5$.
 \end{abstract}

\thispagestyle{empty}
\addtocounter{page}{-1}

\newpage

\section{Introduction}

In the Tree Augmentation Problem (TAP) we are given an undirected tree $T=(V,E)$ and a set of links $L\subseteq \binom{V}{2}$ between pairs of vertices. The task is to find a minimum cardinality set of links $U\subseteq L$ such that $(V,E\cup U)$ is $2$-edge-connected.
A natural extension of TAP is its weighted version, the Weighted Tree Augmentation Problem (WTAP), in which each link
$\ell\in L$ has a nonnegative cost $c_{\ell}$, and the task is to choose a set of links $U\subseteq L$ \emph{of minimum cost} 
such that $(V,E\cup U)$ is $2$-edge-connected.

TAP is a basic Survivable Network Design problem, i.e., a problem where our goal is to enhance the connectivity 
properties of an input network in order to make it tolerant to edge faults. The motivation behind studying such problems 
is clear: networks are not static, but evolve over time, and in particular they are subject to faults that result in 
edges or whole sub-networks becoming unavailable. Therefore, it is desirable that we design networks in such a way that 
the failure of a small number of network interfaces does not disrupt its connectivity. Specifically, TAP 
poses one of the arguably most basic questions in this context: given a network whose connectivity is disrupted even if 
a single link fails, how can we introduce as few links as possible in such a way, that a single link failure will still 
result in a connected network? In this sense, it is a special case of a much more general problem: given a $k$-edge-connected 
network, how can we introduce new links such that the resulting network is $(k+1)$-edge-connected? Interestingly enough, 
whenever $k$ is odd, this problem reduces to TAP (see Cheriyan, Jord{\'a}n, and Ravi~\cite{cheriyan19992}).

Given the importance of the problem, it is not surprising that TAP, WTAP, and even special cases thereof attracted 
considerable interest. To begin with, WTAP was shown to be NP-hard in the 80's by Frederickson and 
J\'aj\'a~\cite{frederickson1981approximation}, and also TAP was proven to be NP-hard later on (see Cheriyan, Jord{\'an}, and Ravi~\cite{cheriyan19992}). Unsurprisingly, considerable effort has been put into designing approximation algorithms.

For WTAP, the best-known approximation guarantee is $2$ and was first established by Frederickson and J\'aj\'a 
\cite{frederickson1981approximation}. Their algorithm was later simplified by Khuller and Thurimella 
\cite{khuller1993approximation}. A $2$-approximation can also be achieved by various other techniques developed later 
on, including a primal-dual technique by Goemans, Goldberg, Plotkin, Shmoys, Tardos, and 
Williamson~\cite{goemans_1994_improved}, and the iterative rounding technique by Jain~\cite{jain_2001_factor}.
For WTAP, improvements on the factor $2$ have only been obtained for restricted cases, including bounded diameter trees (see Cohen and Nutov~\cite{cohen20131+}), and instances where the ratio $\Delta$ between the highest and smallest cost is bounded by a constant (see Adjiashvili~\cite{adjiashvili2017beating}, and Fiorini, Gro\ss, K\"onemann, and Sanit\`a~\cite{fiorini_2018_approximating}). Moreover, very recently, Nutov~\cite{nutov_2017_tree} showed that even if $\Delta=O(\log n)$, where $n$ is the number of vertices, the factor of $2$ can be beaten.

Regarding TAP, the first algorithm beating the approximation guarantee of $2$ is due to Nagamochi~\cite{nagamochi_2003_approximation}, achieving an approximation factor of $1.815+\epsilon$. This factor was subsequently improved to $1.8$ by Even, Feldman, Kortsarz, and Nutov~\cite{even_2009_approximation}.
The best-known approximation guarantee is $1.5$, first obtained by Kortsarz and Nutov~\cite{kortsarz2016simplified}. These results are combinatorial in nature, displaying a contrast with the weighted version of the problem, where the majority of the results come from LP-based algorithms. With respect to dulling this contrast, Korsarz and Nutov~\cite{kortsarz2016lp} provided an LP-based $7/4$-approximation algorithm (the algorithm is combinatorial, but the analysis is LP-based). Moreover, Cheriyan and Gao~\cite{cheriyan_2017_approximating} presented a combinatorial $1.5$-approximation, whose analysis is based on an SDP obtained through Lasserre's hierarchy.
More recently, a $(5/3+\epsilon)$-approximation algorithm that is LP-based was given by Adjiashvili \cite{adjiashvili2017beating}. Fiorini, Gro\ss, K\"onemann, and Sanit\`a~\cite{fiorini_2018_approximating} were able to considerably strengthen Adjiashvili's approach to obtain an LP-based $(1.5+\epsilon)$-approximation by adding to Adjiashvili's LP a strong family of additional constraints, commonly known as $\{0,\frac{1}{2}\}$-Chv\'atal-Gomory constraints. Hence, up to an arbitrarily small constant $\epsilon >0$, which impacts the running time, this algorithm also matches the best factor of $1.5$ achieved by Kortsarz and Nutov~\cite{kortsarz2016simplified}, and Cheriyan and Gao~\cite{cheriyan_2017_approximating}.

Interestingly, despite the recent LP-based and SDP-based progress, relaxations for TAP 
are still badly understood. In 
particular, even the integrality gap of the most natural LP, known as the \emph{cut-LP}, was only very recently shown by 
Nutov~\cite{nutov_2017_tree} to be below $2$ (namely at most $2-\sfrac{2}{15}$), even though the best known lower bound 
is $1.5$, which is due to an example of Cheriyan, Karloff, Khandekar, and K{\"o}nemann~\cite{cheriyan2008integrality}. 
Hence, it is unknown whether the stronger LPs introduced in~\cite{adjiashvili2017beating} 
and~\cite{fiorini_2018_approximating} are any stronger than the canonical cut-LP.
Moreover, the fact that $1.5$-approximations (or $1.5+\epsilon$, respectively) have now been obtained by three 
independent approaches, one combinatorial~\cite{kortsarz2016simplified}, one 
SDP-based~\cite{cheriyan_2017_approximating}, and one LP-based~\cite{fiorini_2018_approximating}, together with the 
difficulty to obtain LPs or SDPs with an integrality gap strictly below $1.5$ (even with a non-algorithmic proof), has 
led to the question whether it is possible to obtain approximation factors strictly below $1.5$ for TAP (see, e.g., 
K\"onemann~\cite{konemann_2017_improved}, for an example where this question was raised explicitly). In this paper, we 
answer this question in the affirmative.

\subsection{Our Results}

Our main result is the first approximation algorithm beating the factor of $1.5$ for TAP.

\begin{theorem}\label{thm:main}
There exists a $1.458$-approximation algorithm for TAP. 
\end{theorem}

A first step in our approach exploits key observations done by Adjiashvili~\cite{adjiashvili2017beating} to reduce to a very structured TAP instance, which we call $k$-wide, and is defined as follows.
\begin{definition}[$k$-wide instance, principal subtree]
Let $k\in \mathbb{Z}_{\geq 1}$. We say that a tree $T=(V,E)$ is \emph{$k$-wide}, if there exists a vertex $r\in V$, which we also call \emph{root}, with the following property.
Any subtree of $T$ consisting of $r$, one of its children $u$, and all descendants of $u$ has at most $k$ leaves; moreover, these subtrees are called \emph{principal subtrees}.
\end{definition}

Starting with the work of Adjiashvili~\cite{adjiashvili2017beating}, several recently developed approaches \cite{fiorini_2018_approximating,nutov_2017_tree} can be interpreted as improving approximation factors and LP formulations for $O(1)$-wide TAP instances, and propagating such improvements to general TAP instances.
In a first step, we provide a black-box reduction of general TAP to $O(1)$-wide TAP instances. In previous approaches based on Adjiashvili's technique, there remained some entanglement between improvements on approximating $O(1)$-wide TAP instances and the propagation of these improvements to general TAP instances. Our statement below avoids such entanglements and allows for a simpler presentation of our main technique.
\begin{theorem}\label{thm:decomp}
Let $k\in \mathbb{Z}_{\geq 1}$, $\alpha \geq 1$, and $\mathcal{A}$ be an algorithm that is an $\alpha$-approximation 
for TAP on $k$-wide instances. Then there is an $(\alpha+O(\sfrac{1}{\sqrt{k}}))$-approximation algorithm for TAP, making a 
polynomial number of calls to $\mathcal{A}$ and performing further operations that take time polynomially bounded in the 
input size.
This reduction also works for randomized algorithms, where the approximation guarantees hold in expectation.
\end{theorem}

Finally, our main contribution, which leads to Theorem~\ref{thm:main} due to Theorem~\ref{thm:decomp}, is a new approximation algorithm for $O(1)$-wide instances, whose existence is formally stated in the following theorem.

\begin{theorem}\label{thm:mainKWide}
There exists an approximation algorithm for $O(1)$-wide TAP instances with approximation guarantee $\sfrac{\sqrt{34}}{4} < 1.458$.
\end{theorem}

For simplicity, the above results are stated only for unweighted instances. However, both Theorem~\ref{thm:decomp} and Theorem~\ref{thm:mainKWide} can be extended, in a weaker form, to WTAP instances where the ratio of the highest cost $c_{\max}$ to the lowest cost $c_{\min}$ is bounded by $\Delta$.
These extensions, on which we expand in Appendix~\ref{sec:WTAP}, show that we can beat the approximation factor of $1.5$ 
even for such instances by a constant that depends on $\Delta$.
 
\subsection{Brief Discussion of Used Techniques}

The black-box reduction presented by Theorem~\ref{thm:decomp} can be obtained by a careful combination of known techniques, including core reduction ideas by Adjiashvili~\cite{adjiashvili2017beating} and the use of the ellipsoid algorithm on a well-defined LP with a ``partial'' separation oracle based on an approximation algorithm for $O(1)$-wide instances. Ellipsoid-based techniques with such partial separation oracles have been employed successfully in prior work, including in the context of facility location~\cite{carr_2000_strengthening, levi_2008_approximation, an_2017_lp-based}, and very recently, Nutov~\cite{nutov_2017_tree} presented an elegant application of it by showing that a certain LP has integrality gap below $2$ for weighted TAP with bounded ratio between highest and lowest cost.

To show our main technical result, i.e., a $\frac{\sqrt{34}}{4}$-approximation for $O(1)$-wide instances as claimed by
Theorem~\ref{thm:mainKWide}, we introduce a new technique to improve solutions, which we call \emph{rewiring}.
More precisely, a common approach in algorithms to round solutions to relaxations of TAP is to \emph{split} certain links into two smaller
ones that cover the same edges. This approach is particularly well-suited for $O(1)$-wide instances where one can break the TAP problem into
independent ones over the principal subtrees, each of which only has $O(1)$-many leaves, by splitting all links that go over the root. We
show that, when starting with an appropriate LP, then after independently rounding the subtrees, one can further improve by identifying
pairs of links across different principal subtrees that can be replaced by a single link. We show that this rewiring process can be reduced
to a certain type of matching problem. We reduce the question of whether one can make a substantial improvement through rewirings, to a
question about the expected cardinality of maximum matchings in a certain random graph, for which we can show that large matchings always
exist in expectation.

\paragraph{Outline of the Paper.}
Section~\ref{sec:wide} proves Theorem~\ref{thm:mainKWide} by presenting our approximation algorithm for $O(1)$-wide TAP 
instances. To better highlight our main ideas, we first present a randomized approximation algorithm with an 
approximation factor of $1.465$, and expand on the derandomization in Appendix~\ref{sec:derandomization}, and on the 
better approximation factor stated in Theorem~\ref{thm:mainKWide} in Appendix~\ref{sec:improved-apx}. The black-box 
reduction to $k$-wide instances is contained in Appendix~\ref{sec:reductionToWideTrees}, modulo one technical result 
that closely follows a derivation by Adjiashvili~\cite{adjiashvili2017beating}. Due to space constraints, we defer a 
formal proof of this result to the long version of the paper.
Finally, Appendix~\ref{sec:calculus} contains some relatively standard but technical proofs that help us bound the 
approximation ratio of our procedure for $k$-wide instances.  
Finally, Appendix~\ref{sec:WTAP} contains a discussion on how to extend our results to weighted TAP instances with a 
bounded ratio between maximum and minimum link cost.

\paragraph{Basic Preliminaries.}

Consider a TAP instance $(T,L)$ on a tree $T=(V,E)$ with links $L\subseteq \binom{V}{2}$. For a link $\ell\in L$ we denote by $P_{\ell}\subseteq E$ the unique path between the endpoints of $\ell$ in $T$. Hence, TAP can be rephrased as the following covering problem
\begin{equation*}
\min\left\{|U| \;\middle\vert\; U\subseteq L, E=\cup_{\ell\in U}P_{\ell}\right\}\enspace.
\end{equation*}
For an edge $e\in E$, we denote by $\cov(e)=\{\ell\in L \mid e\in P_{\ell}\}$ the set of all links covering $e$. More generally, for $U\subseteq E$, we denote by $\cov(U)=\cup_{e\in U}\cov(e)$ all links that cover at least one edge of $U$.
The classical \emph{cut-LP} seeks to minimize $x(L)\coloneqq \sum_{\ell\in L}x_\ell$ over all $x$ in the polytope
\begin{equation*}
\Pi_T \coloneqq \left\{ x\in [0,1]^L  \;\middle\vert\; x(\cov(e)) \geq 1 \;\; \forall e\in E\right\}\enspace.
\end{equation*}

Moreover, for an integer $k\in \mathbb{Z}_{\geq 1}$, we use the notation $[k]\coloneqq \{1,\ldots, k\}$.

\section{Approximating $k$-Wide TAP}\label{sec:wide}

In this section, we prove Theorem~\ref{thm:mainKWide} by presenting our approximation algorithm for TAP on $k$-wide instances. We start in Section~\ref{sec:lp} by presenting a novel efficiently solvable LP relaxation for $k$-wide TAP instances. In Section~\ref{sec:algo}, we present a rounding algorithm for its solution, which is based on performing rewirings.
Interestingly, the number of rewirings can be lower bounded by the expected size of a matching on a particular type of random graph. In Section \ref{sec:matching}, we show that the latter expected size is large enough, hence implying our improved approximation factor as discussed in Section \ref{sec:analysis}.
Incidentally, this also shows that our LP relaxation has an integrality gap below $1.5$ for $k$-wide instances. This is the first LP relaxation which is known to have this property.

Throughout this section, we assume that we are dealing with TAP instances $(T,L)$ that are \emph{shadow-complete}. This means that for every link $\ell\in L$, and any two vertices $u,v$ on the path $P_\ell$, we have $\{u,v\}\in L$. In this case we have $P_{\{u,v\}}\subseteq P_{\ell}$, and $\{u,v\}$ is called a \emph{shadow} of $\ell$. Clearly, one can assume without loss of generality that a TAP instance is shadow-complete by introducing all shadows of links. This leads to an equivalent problem because any solution using a shadow $\ell'$ of an original link $\ell$ can be modified to another solution by replacing $\ell'$ by $\ell$, without increasing the number of links in the solution.

Moreover, we will consider different rounding algorithms whose approximation guarantees depend on the type of links involved. Given a $k$-wide TAP instance $(T,L)$ with root $r$, we partition $L$ into $\Lcross\cup \Lin$, where $\Lcross$, called \emph{cross-links}, are all the links whose endpoints are different from the root and are in different principal subtrees. The remaining set $\Lin$, called \emph{in-links}, are all the links with both endpoints in the same principal subtree. Moreover, $\Lup\subseteq\Lin$, called \emph{up-links}, are all the in-links $\{u,v\}\in\Lin$ such that one of the endpoints of $\{u,v\}$, say $u$, lies on the path from $r$ to $v$. We sometimes also use the notions of cross-links, in-links, and up-links with respect to some vertex $r$ for instances that are not $k$-wide.

\subsection{A New LP Relaxation for $k$-Wide TAP}
\label{sec:lp}

Recall that the cut-LP has integrality gap at least $1.5$~\cite{cheriyan2008integrality}.\footnote{Notice, also for $k$-wide instances, the cut-LP can have an integrality gap arbitrarily close to $1.5$ (for large enough $k$). This follows by the fact that a principal subtree can be any TAP instance with at most $k$ leaves. Hence, for large enough $k$, any integrality gap instance is a $k$-wide TAP.}
We thus need to strengthen it to beat this factor.
We do this by adding two sets of constraints.
We start with $\{0,\frac{1}{2}\}$-Chv\'atal-Gomory cuts, which were introduced in the context of TAP by Fiorini et al.~\cite{fiorini_2018_approximating}. They are described by the following polytope: 
\begin{equation*}
\Pi_{T}^{\CG} = \left\{
x\in [0,1]^L \;\middle\vert\;
x(\pi(S))\geq \frac{|\delta_E(S)|+1}{2}\quad \forall S\subseteq V \text{ with } |\delta_E(S)| \text{ odd}
\right\}\enspace,
\end{equation*}
where $\delta_E(S)$ are all edges with precisely one endpoint in $S$, and $\pi(S)$ is the multiset of links whose corresponding paths intersect $\delta_E(S)$, where the multiplicity of any link $\ell$ is $\lceil | P_{\ell} 
\cap 
\delta_E(S)|/2\rceil$. These constraints imply the cut-LP constraints.\footnote{The fact that the cut-LP constraint for any edge $e\in T$ is implied by $\Pi_T^{\CG}$, and hence $\Pi_T^{\CG}\subseteq \Pi_T$, follows by considering the set $S$ corresponding to the vertex set of one of the two connected components of $T$ when removing the edge $e$.}
Moreover, one can efficiently separate over $\Pi_T^{\CG}$ (see~\cite{fiorini_2018_approximating}).
A key property of $\Pi_T^{\CG}$ shown by Fiorini et al.~\cite{fiorini_2018_approximating}, is that points in $\Pi_T^{\CG}$ can be rounded losslessly on a particular type of instance, as summarized below.

\begin{theorem}[\cite{fiorini_2018_approximating}]\label{thm:fioriniGCexact}
Let $(T,L)$ be a TAP instance, let $x\in \Pi_T^{\CG}$, and let $r\in V$. If every link $\ell\in \supp(x)$ is either
\begin{enumerate*}[label=(\roman*)]
\item an uplink with respect to $r$, or
\item a cross-link with respect to $r$,
\end{enumerate*}
then one can losslessly round $x$, i.e., a solution $R\subseteq L$ to the TAP problem with $|R|\leq x(L)$ can be efficiently obtained. This even holds for WTAP.
\end{theorem}

The above theorem leads to a natural rounding algorithm for a point $x\in \Pi_T^{\CG}$. This rounding procedure, which was suggested in~\cite{fiorini_2018_approximating}, works as follows when applied to the root of a $k$-wide TAP instance $T=(V,E)$ with links $L$. We first modify $x$ to avoid the use of in-links that are not up-links, which will bring us to the setting of Theorem~\ref{thm:fioriniGCexact}.
More precisely, for $\ell=\{u,v\}\in \Lin\setminus \Lup$, let $a\in V$ be the vertex closest to the root among the vertices on the path $P_{\ell}$. Then one can modify $x$ by increasing its value on $\{u,a\}$ and $\{v,a\}$ by $x(\ell)$, and setting the value of $x(\ell)$ to $0$. Doing this for all links in $\Lin\setminus \Lup$, a new point $y\in \Pi_T^{\CG}$ is obtained with support only on cross-links and up-links and such that $y(\Lcross)=x(\Lcross)$ and $y(\Lup) = 2x(\Lin) - x(\Lup)$. Finally, by applying Theorem~\ref{thm:fioriniGCexact} to $y$, one obtains the following.

\begin{corollary}[\cite{fiorini_2018_approximating}]\label{cor:chvatal-rounding}
Given a point $x\in \Pi_T^{\CG}$, we can compute in polynomial time a solution $\mathcal{A}$ such that
\begin{equation*}
|\mathcal{A}|\leq 2x(\Lin)-x(\Lup)+x(\Lcross)\enspace.
\end{equation*}
\end{corollary}

On top of the Chv\'atal-Gomory constraints, we introduce a second set of constraints that deviates from prior work, but can be interpreted as a strengthening of so-called \emph{bundle constraints} introduced in~\cite{adjiashvili2017beating}, which were also used (with non-trivial technical adjustments) in~\cite{fiorini_2018_approximating,nutov_2017_tree}. Bundle constraints make sure that when restricting an LP solution to any subtree of $T$ with a constant number of leaves, then the value of this restricted LP solution is not lower than that of an optimal solution covering the same subtree.
Being able to losslessly round an LP solution on subtrees with $k=O(1)$ leaves, by replacing the LP solution with an integral solution of no higher cost, is of particular interest in the context of $k$-wide TAP instances due to the following. One can first independently consider each principal subtrees of the $k$-wide instance and round it, and then take the union of the links used in the solutions to the different principal subtrees.
As first shown in~\cite{adjiashvili2017beating}, this allows for obtaining a solution with at most $x(\Lin)+2x(\Lcross)$ links. By returning the better of this solution and the one guaranteed by Corollary~\ref{cor:chvatal-rounding}, a $1.5$-approximation is obtained for $k$-wide instances. Finally, by Theorem~\ref{thm:decomp}, this implies a $(1.5+\epsilon)$-approximation for any TAP instance. We highlight that this is a nice example of how Theorem~\ref{thm:decomp} allows for providing a more elegant description of existing techniques, because we can restrict ourselves to $k$-wide instances.

We will provide techniques that improve on this rounding procedure. More precisely, our main goal now is to introduce additional constraints such that if one rounds each principal subtree independently, then there is a way to further improve the solution after we take the union of all rounded links. For this, we want to be able to round each principal subtree in a well-structured way. More precisely, we first need to be able to decompose the LP solution, when restricted to one principal subtree, into a \emph{convex combination of integral solutions}. Moreover, we want that the links used in solutions appearing in this convex combination do not contain redundancies, which we formalize below through the notion of \emph{shadow-minimality} for link families.
\begin{definition}
Consider a TAP instance $(T,L)$.
Two links $\ell_1, \ell_2 \in L$ are \emph{shadow-minimal} if there exists no shadow $s$ of one link, say $\ell_1$, such that $P_{\ell_1} \cup P_{\ell_2} = P_s \cup P_{\ell_2}$. Moreover, a set of links $L'\subseteq L$ is \emph{shadow-minimal} if all links in $L'$ are pairwise shadow-minimal.
\end{definition}

The following basic observation shows that requiring shadow-minimality is not restrictive. In other words, there always exists an optimal solution that is shadow-minimal. This follows by the fact that, whenever we have a non-shadow-minimal solution, then some link can be shortened to one of its shadows without destroying feasibility.
\begin{observation}\label{obs:shadow-minimalOPT}
 Given a TAP instance $(T,L)$, and given some $L_1\subseteq L$, there exists a shadow-minimal set of links $L_2\subseteq L$ such that $|L_1|\geq|L_2|$, and $L_2$ covers the same edges as $L_1$. Furthermore, $L_2$ contains only links from $L_1$, or 
shadows thereof.
\end{observation}

To describe our constraint set, which allows for rounding principal subtrees in a shadow-minimal way, we present here a compact extended formulation for simplicity. It is also possible to achieve the same goal via an exponential-size set of constraints (and no extra variables). 

Fix a principal subtree $T_i = (V_i,E_i)$, $i\in [q]$. Conceptually, one could introduce a binary decision variable $y_S$ for any set $S\subseteq \cov(E_i)$ that covers all the edges of $T_i$ and is shadow-minimal, and then enforce the constraint of picking exactly one such set $S$ for each $T_i$. Clearly, there are exponentially many such sets $S$. We will show in the following that it suffices to focus on a polynomially small family of sets. Interestingly, even though we use shadow-minimality primarily for the rounding procedure to be introduced later, it also allows us to restrict the number of sets $S$ we have to consider, thus obtaining a family of polynomial size. To do so, we start by observing that shadow-minimal solutions of $k$-wide trees contain few cross-links:

\begin{lemma}\label{lem:minimal-cross}
	 Given a TAP instance $(T,L)$, where $T$ is a $k$-wide tree, consider a shadow-minimal set $L'\subseteq L$. Then, for each principal subtree $T_i$, $L'$ contains at most $k$ cross-links with an endpoint inside $T_i$.
\end{lemma}
\begin{proof}
	Suppose towards contradiction that there exists a principal subtree $T_i$ such that at least $k+1$ cross-links in $L'$ have an endpoint inside $T_i$. Since $T_i$ contains at most $k$ leaves, we have by the pigeonhole principle that there is one leaf $u$ such that the path from $u$ to $r$ in $T$ contains the endpoints of at least two cross-links, which contradicts the shadow-minimality of $L'$. 
\end{proof}

Let $\Lambda_i\subseteq 2^{\cov(E_i)}$ be the family of all shadow-minimal subsets of all cross-links with one endpoint in $V_i$. By Lemma~\ref{lem:minimal-cross}, each set $R\in \Lambda_i$ satisfies $|R|\leq k$. Hence, $\Lambda_i$ has polynomial size; more precisely, $|\Lambda_i| = O(|L|^k)$. For each $R\in \Lambda_i$, let $C(i,R)\subseteq L$ be a minimum cardinality set of links with both endpoints in $V_i$ that satisfies 
\begin{enumerate}[nosep,label=(\roman*),topsep=0.2em]
\item $R\cup C(i,R)$ is shadow-minimal, and
\item $R\cup C(i,R)$ covers $E_i$.
\end{enumerate}
Notice that such a set $C(i,R)$ can indeed be computed efficiently. The requirement that $R\cup C(i,R)$ covers $E_i$ corresponds to $C(i,R)$ being a TAP solution to the residual instance consisting of the tree $T_i$ after contracting the edges already covered by $R$. However, being a subtree of the $k$-wide tree $T_i$, this residual instance has no more than $k$ leaves, and TAP instances with constantly many leaves are known to be efficiently solvable (see, e.g.,~\cite{adjiashvili2017beating,nutov_2017_tree}).
To guarantee shadow-minimality of $R\cup C(i,R)$, one can first only consider links in the residual instance that do not lead to a violation of shadow-minimality with $R$, i.e., we delete in the residual instance any link $\ell$ such that $R\cup \{\ell\}$ is not shadow-minimal. Finally, any optimal solution to the residual instance that does not use any link $\ell$ such that $R\cup \{\ell\}$ is not shadow-minimal, can easily be transformed into a shadow-minimal one by shortening links if necessary.
Let $L_i^R = R\cup C(i,R)$.

Due to Observation~\ref{obs:shadow-minimalOPT} and Lemma~\ref{lem:minimal-cross}, there is an optimal solution to the $k$-wide TAP instance we consider that is of the form $L_1^{R_1} \cup \ldots \cup L_q^{R_q}$, where $R_i \in \Lambda_i$ for $i\in [q]$.
Indeed, consider some shadow-minimal optimal solution $L^*$; due to Observation \ref{obs:shadow-minimalOPT}, there exists at least one. Due to Lemma \ref{lem:minimal-cross}, we know $L^*$ contains at most $k$ cross-links with an endpoint inside any principal subtree $T_i$. Let $L^*_i$ be the cross-links of $L^*$ that have an endpoint in $T_i$; then $\bigcup_{i=1}^{q} (L^*_i \cup C(i,L^*_{i}))$ is an optimal solution of the desired form.

Our new LP constraints are based on this observation, and we have a variable $\lambda_i^R\geq 0$ for each principal subtree $T_i$ and each set $R\subseteq \Lambda_i$, where $\lambda_i^R=1$ is interpreted as including the links $L_i^R$ in the solution.
The local convex decomposition constraints we introduce ensure that $x$ and $\lambda$ are consistent, and that we only choose 
one set of links for each one of the $q$ principal subtrees in $T$:
\begin{equation}\label{eq:convDecConstraints}
\begin{array}{>{\displaystyle}rc>{\displaystyle}l@{\qquad}l}
x_\ell&=&\sum_{R\in\Lambda_i:\ell\in L_i^R}\lambda_{i}^{R},
  &\forall i\in [q]  \quad\forall \ell\in \cov(E_i),\\
\sum_{R\in\Lambda_i}\lambda_{i}^{R}&=&1, &\forall i\in [q],\\
\lambda_i^R &\geq &0 &\forall i \in [q] \quad \forall R\in \Lambda_i\enspace.
\end{array}
\end{equation}

As discussed, there must be at least one $x\in \{0,1\}^L$ that is the characteristic vector of an optimal solution and is feasible for the above constraints. Hence, the above constraints can be interpreted as a relaxation of the original TAP problem.

Finally, the feasible solutions of our LP consist of all tuples $(x,\lambda)$, where $x\in \mathbb{R}^L$ and $\lambda$ has $\sum_{i=1}^q |\Lambda_i|$ many components, one for each $i\in [q]$ and $R\in \Lambda_i$, such that $x\in \Pi_T^{\CG}$ and $(x,\lambda)$ satisfies~\eqref{eq:convDecConstraints}. Moreover, we minimize $x(L)$ over this polytope. For brevity, we call this resulting LP, the \emph{$k$-wide-LP}. One can efficiently optimize over the $k$-wide-LP for any $k=O(1)$, because there is an efficient separation oracle for $\Pi^{\CG}_T$, and there are polynomially many additional constraints and variables described by~\eqref{eq:convDecConstraints}.

We summarize that the key property we gain through our $k$-wide-LP compared to previous linear programs, is that any solution $x$ to the $k$-wide-LP, when restricted to the links $\cov(E_i)$ for any $i\in [q]$, can be written as a convex combination of shadow-minimal link sets, each covering $E_i$.

\subsection{The Rounding Algorithm}
\label{sec:algo}

We next present our rounding algorithm based on rewirings. Let $(x,\lambda)$ be an optimal fractional solution to the $k$-wide-LP. We return the better of two solutions obtained by two different rounding procedures. The first procedure is the one by Fiorini et al.~\cite{fiorini_2018_approximating}, whose guarantee is stated in Corollary~\ref{cor:chvatal-rounding}.
The second rounding algorithm deviates substantially from prior work, and is one of the main algorithmic contributions of this paper. For simplicity, we first present a randomized version of our rounding algorithm, and discuss in Appendix~\ref{sec:derandomization} how it can be derandomized.

Note that the links in $x$ that cover the edges of each $T_i$ are expressed as a convex combination of integral solutions $L_i^R$. We will interpret the coefficients $\lambda_i^R$ of this convex combination as a probability distribution and sample from it, independently for each $T_i$. Let $L_i$ be the (random) local solution for each $T_i$, and we start by considering $\cup_{i=1}^q L_i$, which is clearly a feasible integral solution.
Notice that each in-link $\ell$ has both endpoints in one principal subtree $T_i$, and is contained in $L_i$ with probability $x_{\ell}$ (and not contained in any other $L_j$). On the other hand, each cross-link $\ell$ has its endpoints in two different subtrees, say $T_i$ and $T_j$, and is thus contained in each of $L_i$ and $L_j$ with probability $x_{\ell}$ (and it is not contained in any other $L_h$). This simple line of argumentation leads to the following guarantee.
\begin{equation}\label{eq:boundUnionLi}
\E\left[|\cup_{i=1}^q L_i|\right] \leq \E\left[\sum_{i=1}^q |L_i|\right] = x(\Lin) + 2x(\Lcross)\enspace,
\end{equation}
which is the same guarantee that was also obtained by 
Adjiashvili~\cite{adjiashvili2017beating} through a slightly different approach.

In order to improve on this, we introduce our rewiring process that will replace some pairs of cross-links appearing in different principal subtrees by a single one. To do such a replacement, we introduce the notion of \emph{active vertices}, which will be potential endpoints of a new link that allows for removing two existing ones.

\begin{definition}\label{def:activeVert}
Let $u\in V$, and let $i\in [q]$ such that $u\in V_i$. Then $u$ is called \emph{active} if there exists a cross-link $\ell_u\in L_i\cap \Lcross$ that has $u$ as one of its endpoints.
\end{definition}
Notice that for an active vertex $u$, the corresponding link $\ell_u$, as described in the definition, is unique; for otherwise the links $L_i$ will not be shadow-minimal.

Our rewiring step is based on the following observation of how one can improve the solution $\cup_{i=1}^q L_i$. Whenever there are two active vertices $u,v$, say in subtrees $T_i$ and $T_j$, respectively, such that $\{u,v\}\in \Lcross$, then we can replace both the cross-link $\ell_u$ in $L_i$ that has $u$ as one of its endpoints and the cross-link $\ell_v$ in $L_j$ that has $v$ as one of its endpoints, by the single cross link $\{u,v\}$. The following proposition formalizes this operation for multiple pairs $\{u,v\}$ forming a matching $M$.
\begin{proposition}\label{prop:rewiringIsOk}
Let $M\subseteq \Lcross$ be a matching between active vertices. Moreover, for each active vertex $u$, which is part of some subtree $T_i$,
let $\ell_u\in L_i$ be the unique cross-link in $L_i$ with $u$ as one of its endpoints, and we denote by $L_i(M)$ the set of all $\ell_u\in
L_i$ for $u\in V_i$ being an endpoint of some edge in $M$. Then
\begin{equation*}
\mathcal{B}\coloneqq \left(\bigcup_{i=1}^q \left(L_i\setminus L_i(M)\right)
\right) \cup M
\end{equation*}
is a solution to the considered $k$-wide TAP instance. 
\end{proposition}
\begin{proof}
We check that all edges in each principal subtree are covered. Hence, let $i\in[q]$, and consider an edge $e\in E_i$. We recall that by construction, $L_i$ covers all edges $E_i$.
If $e$ is covered by a link in $L_i\setminus L_i(M)$, then it is clearly covered by the same link in $\mathcal{B}$. Otherwise, if $e$ is covered by a link $\ell\in L_i(M)$, then there is a cross-link in $M$ that covers all edges within $E_i$ that $\ell$ covered, including the edge $e$.
\end{proof}

Notice that the cardinality of the solution $\mathcal{B}$ described in Proposition~\ref{prop:rewiringIsOk} can be bounded by
\begin{equation}\label{eq:boundCardCalL}
|\mathcal{B}| \leq \left(\sum_{i=1}^q |L_i\setminus L_i(M)|\right) + |M|
    = \left(\sum_{i=1}^q|L_i|\right) - |M| \enspace,
\end{equation}
which leads to an improvement compared to~\eqref{eq:boundUnionLi} (after taking expectations).
We call the process of replacing a pair of links $\ell_{u},\ell_{v}$ with $\{u,v\}\in \Lcross$ by a single link $\{u,v\}$ a \emph{rewiring}.
It may happen that $\ell_{u}$ and $\ell_{v}$ are the same. Nevertheless, we still improve on the analysis presented in~\eqref{eq:boundUnionLi}, which counts these links twice. However, it is crucial that we do not try to improve on the analysis presented in~\eqref{eq:boundUnionLi} by only exploiting that we improve if the same cross-link appears in two different $L_i$. Such events can be very unlikely, and would not be enough to improve on existing methods. Hence, we do need to do non-trivial rewirings.

Maximizing the number of rewirings $h$ can be achieved by solving a matching problem over the graph consisting of all active vertices, and whose edges are all cross-links between active vertices.
However, it turns out that we only need to consider rewirings on a subset of the active vertices, namely those of degree different from $2$. This is sufficient for our analysis and simplifies the exposition. Hence, let $\Vcr \subseteq V\setminus \{r\}$ be the set of all vertices, except for the root, with degree different from $2$, and let $A\subseteq \Vcr$ be the random set corresponding to all active vertices in $\Vcr$. For brevity, we call the vertices in $\Vcr$ \emph{critical vertices}.
Moreover, we partition the cross-links $\Lcross$ into the links $\Lcr$ with both endpoints being critical, and the remaining links $\Lncr$. Hence, our rewiring procedure finds a maximum cardinality matching in the graph
$(A,\Lcr\cap
\begin{psmallmatrix}
A\\
2
\end{psmallmatrix})$,
say with matching edges $M\subseteq \Lcross$, and returns the solution $\mathcal{B}$ as described in Proposition~\ref{prop:rewiringIsOk}. We highlight that in our description so far, our rounding algorithm is randomized due to the random choice of $L_1,\ldots, L_q$, which also leads to a random set $A$ of active critical vertices. To make sure that we can find a sufficiently large matching, it is crucial to observe that the distribution of $A$ has the following properties:
\begin{equation}\label{eq:propsOfA}
\begin{alignedat}{2}
\Pr[v\in A] &= \sum_{\ell\in \Lcr, v\in \ell} x_{\ell} \qquad &&\forall v\in \Vcr\enspace, \text{ and}\\
\Pr[u\in A \text{ and } v\in A] &= \Pr[u\in A] \cdot \Pr[v\in A] \qquad &&\forall \{u,v\}\in \Lcr\enspace,
\end{alignedat}
\end{equation}
where the first property follows from the fact that the link sets $L_i$ are shadow-minimal, which implies that there is at most one cross-link in $L_i$ that has $v$ as its endpoint. Moreover, the second property follows from the fact that if $\{u,v\}\in \Lcr$, then $u$ and $v$ are in different principal subtrees, say $T_i$ and $T_j$, respectively. Indeed, the independence of $\{u\in A\}$ and $\{v\in A\}$ then follows by the independence of the distributions of $L_i$ and $L_j$, which holds because we round each principal subtree independently.
The following lemma summarizes the key properties of the solution we obtain through our rounding algorithm based on rewiring. For any graph $G=(V,E)$, we denote by $\eta(V,E)$ the cardinality of a maximum cardinality matching in $G$.

\begin{lemma}\label{lem:wired-rounding}
Consider a $k$-wide TAP instance $(T,L)$ and a solution $(x,\lambda)$ to the $k$-wide-LP. There exists a randomized algorithm that returns a solution $\mathcal{B}\subseteq L$ to the TAP instance $(T,L)$, such that
$$
\E[|\mathcal{B}|] \leq x(\Lin)+2x(\Lcross)-\E_A\left[ \eta\left(A,\Lcr\cap
\begin{psmallmatrix}
A\\
2
\end{psmallmatrix}
\right)\right]\enspace,
$$
where $A$ is a random subset of $\Vcr$ with a distribution satisfying~\eqref{eq:propsOfA}.
\end{lemma}
\begin{proof}
The set $A$ corresponds to the active vertices that are critical, which fulfill~\eqref{eq:propsOfA}. Moreover,
\begin{equation*}
\E[|\mathcal{B}|] \leq \E\left[ \textstyle\sum_{i=1}^q |L_i| \right]
- \E_A\left[\eta\left(A,\Lcr\cap
\begin{psmallmatrix}
A\\
2
\end{psmallmatrix}
\right)\right]
=  x(\Lin) + 2x(\Lcross)  
- \E_A\left[\eta\left(A,\Lcr\cap
\begin{psmallmatrix}
A\\
2
\end{psmallmatrix}
\right)\right]\enspace,
\end{equation*}
where the inequality follows from~\eqref{eq:boundCardCalL} and the equality from~\eqref{eq:boundUnionLi}.
\end{proof}

\subsection{Lower Bounding the Matching Size}
\label{sec:matching}
The next lemma shows that
$\E_{A}[\eta(A, \Lcr \cap
\begin{psmallmatrix}
A\\
2
\end{psmallmatrix})]$
is large if $x(\Lcr)$ is so.
\begin{lemma}\label{lem:matching}
	Let $G=(V,E)$ be an undirected graph, let $x\in\mathbb{R}_{\geq 0}^E$ such that for any $v\in V$ we have $x(\delta(v)) \leq 1$, and let $A$ be a random subset of $V$ with a distribution that satisfies:
	\begin{enumerate}[nosep,label=(\roman*),itemsep=0.2em,topsep=0.2em]
		\item\label{sample1}  $\Pr[v\in A]=x(\delta(v)) \quad \forall v\in V$, and 
		\item\label{sample2}  $\Pr[v\in A\text{ and }u\in A]= \Pr[u\in A] \cdot \Pr[v\in A]\quad \forall
\{u,v\}\in E$.
	\end{enumerate} 
Then there exists a matching $M\subseteq E$ in $G$ such that $\E\left[ |M\cap
\begin{psmallmatrix}
A\\
2
\end{psmallmatrix}
|  \right] \geq \sfrac{x(E)^2}{|V|}$.
\end{lemma}
Before we prove the lemma, observe that
$M\cap \begin{psmallmatrix}
A\\
2
\end{psmallmatrix}$
is clearly a matching in
$(A,E\cap \begin{psmallmatrix}A\\ 2\end{psmallmatrix})$,
because $M$ is a matching. The above lemma thus implies that
$\E\left[\eta\left(A,E\cap
\begin{psmallmatrix}
A\\
2
\end{psmallmatrix}
\right)\right] \geq \sfrac{x(E)^2}{|V|}$.

\begin{proof}

We start by replacing $x$ by a sparser point $y$. For this, consider the polytope
\begin{equation*}
Q = \left\{y\in \mathbb{R}^E_{\geq 0} \;\middle\vert\; y(\delta(v)) = x(\delta(v)) \quad \forall v\in V\right\}\enspace,
\end{equation*}
which is non-empty because $x\in Q$. Let $y$ be any vertex of $Q$. We have $|\supp(y)|\leq |V|$, which follows by the following standard sparsity argument:\footnote{Recall that the support $\supp(y)\subseteq E$ of $y$ are all edges $e\in E$ with $y(e)>0$.} Any vertex $y$ of $Q$ is defined by $|E|$-many linearly independent and tight constraints; however, $Q$ has only $|V|$ constraints that are not nonnegativity constraints, and therefore, at least $|E|-|V|$ of the nonnegativity constraints must be tight at $y$, which implies $|\supp(y)| \leq |V|$.

We let the matching $M\subseteq E$ be a matching obtained by the greedy algorithm for maximum weight matchings with respect to the weights $y$. More formally, we order the edges $E=\{e_1,\ldots, e_m\}$ such that $y(e_1)\geq \ldots \geq y(e_m)$. We start with $M=\emptyset$ and go through the edges in the order $e_1,\ldots, e_m$. When considering $e_i$, we set $M=M\cup \{e_i\}$ if $M\cup \{e_i\}$ is a matching. We have the following (which holds for any matching $M$):
\begin{equation}\label{eq:expMatchSize1}
\E\left[|M\cap 
\begin{psmallmatrix}
A\\
2
\end{psmallmatrix}
|\right]
= \sum_{\{u,v\}\in M} x(\delta(u)) \cdot x(\delta(v))
= \sum_{\{u,v\}\in M} y(\delta(u)) \cdot y(\delta(v)) \enspace,
\end{equation}
which is an immediate consequence of property~\ref{sample1} and~\ref{sample2}, and of $y\in Q$. Next, we show that the following holds:
\begin{equation}\label{eq:expMatchSize2}
\sum_{\{u,v\}\in M} y(\delta(u)) \cdot y(\delta(v)) \geq \sum_{e\in E} y(e)^2\enspace.
\end{equation}
Note that when expanding the left-hand side of~\eqref{eq:expMatchSize2} by using 
$y(\delta(u))=\sum_{e\in \delta(u)}y(e)$ and $y(\delta(v))=\sum_{e\in \delta(v)}y(e)$, we obtain a sum of terms of the 
form $y(e)\cdot y(f)$ for some pairs $e,f\in E$. For any $e=\{u,v\}\in M$, the left-hand side 
of~\eqref{eq:expMatchSize2} has a term of the form $y(e)^2$, due to the term $y(\delta(u))\cdot y(\delta(v))$. Now 
consider any non-matching edge $e=\{u,v\}\in E\setminus M$. Because $e$ is a non-matching edge, and we constructed $M$ 
by the greedy algorithm, there is at least one endpoint of $e$, say $u$, such that there is an edge $f\in M$ incident 
with $u$ that satisfies $y(f)\geq y(e)$. This implies that $y(\delta(u))\cdot y(\delta(v))$ contains a term $y(f)\cdot 
y(e) \geq y(e)^2$. Hence, for each $e\in E$, we identified a term on the left-hand side of~\eqref{eq:expMatchSize2} of 
value at least $y(e)^2$; moreover, all such terms on the left-hand side are different, which 
implies~\eqref{eq:expMatchSize2}.

The statement of the lemma now follows by combining~\eqref{eq:expMatchSize1} together with~\eqref{eq:expMatchSize2} and the following inequality,
\begin{equation*}
x(E)^2 = y(E)^2 \leq |\supp(y)|\cdot \sum_{e\in E} y(e)^2 \leq |V| \cdot \sum_{e\in E} y(e)^2\enspace,
\end{equation*}
where the equality follows from $y\in Q$, the first inequality follows by the Cauchy-Schwarz inequality, and the second one by $|\supp(y)|\leq |V|$.
\end{proof}

Notice that the conclusion of the lemma is essentially optimal due to the following. Consider a graph $G=(V,E)$ that is a star with center $c\in V$, and let $x\in \mathbb{R}^E_{\geq 0}$ be the vector with value $\sfrac{1}{(|V|-1)}$ on each edge. Moreover, let $A$ be the random set of vertices such that with probability $\sfrac{1}{(|V|-1)}$ we have $A=V$, and with probability $\sfrac{(|V|-2)}{(|V|-1)}$ we have $A=\{c\}$. In this case, $\E[\nu(A,E\cap
\begin{psmallmatrix}
A\\
2
\end{psmallmatrix}
)] = \sfrac{1}{(n-1)}$ and $\sfrac{x(E)^2}{|V|} = \sfrac{1}{n}$.
Clearly, the above example can also be extended to larger values of $x(E)$ by considering disjoint stars.
We highlight that this does not exclude the possibility that, in the context we use the lemma, stronger conclusions may be obtained through a different statement and analysis.

\subsection{The Approximation Factor}
\label{sec:analysis}

We next upper bound the approximation factor of the considered rounding algorithm. For simplicity, and due to space constraints, we present here a slightly weaker upper bound and defer to Appendix~\ref{sec:improved-apx} the refined analysis leading to the bound claimed in Theorem~\ref{thm:mainKWide}.

Consider an optimal fractional solution  $(x,\lambda)$ to the $k$-wide-LP, and let $\OPT^*\coloneqq x(L)$ be its value.  Let $\across\in [0,1]$ be the fraction of $\OPT^*$ due to $x(\Lcross)$, i.e., $\across = x(\Lcross)/\OPT^*$, and define similarly $\acr$, $\ancr$, 
$\ain$ and $\aup$ w.r.t. $\Lcr$, $\Lncr$, $\Lin$ and $\Lup$, respectively. 
Observe that $\aup\leq \ain$,  $\acr+\ancr=\across$, and 
$\across+\ain=1$.

Previously, we showed that our rewiring procedure is the stronger the larger $x(\Lcr)$ 
is. Observe that if $x(\Lcross)$ is small enough, namely significantly less than $\sfrac{x(L)}{2}$, then the rounding 
guarantee of Lemma~\ref{lem:wired-rounding} already improves on the factor $1.5$ even without any rewiring. 
Still it may happen that $x(\Lcross)$ is large and $x(\Lcr)$ is small, which implies that $x(\Lncr)$ is large. The next lemma shows that in this case also $x(\Lup)$ has to be large, which in turn implies that we get a stronger guarantee via Corollary~\ref{cor:chvatal-rounding}.

\begin{lemma}\label{lem:up}
Given a $k$-wide TAP instance $(T,L)$ and a solution $(x,\lambda)$ to the $k$-wide-LP, 
we have $x(\Lup) \geq x(\Lncr)$, which is equivalent to $\aup\geq \ancr$.
\end{lemma}

\begin{proof}
Consider any link $\ell\in \Lncr$, and let $v$ be any one of its non-critical endpoints. Let $T_i$ be the principal subtree containing $v$, and, for any $S\subseteq L$, we denote by $x_S$ be vector $x$ restricted on the coordinates corresponding to $S$.
From the design of the $k$-wide-LP, $x_{\cov(T_i)}$ is a convex combination of shadow-minimal link sets $L_i^R$, each covering $T_i$. In each such solution $L_i^R$ that contains $\ell$, there must exist another link $\ell'$ 
incident on $v$ and covering the edge right below $v$. The link $\ell'$ cannot be in $\Lcross$ or $\Lin\setminus \Lup$, because in both cases shadow-minimality of $\{\ell,\ell'\}$ would be contradicted, since one could shorten $\ell$. Thus $\ell'\in \Lup$. We can thus conclude that
\begin{equation*}
\sum_{\ell\in \Lup:v\in \ell}x_\ell\geq \sum_{\ell\in \Lncr:v\in \ell}x_{\ell}\enspace.
\end{equation*}
The claim follows by summing over all non-critical nodes $v$.
\end{proof}

\begin{lemma}\label{lem:apx}
There exists a randomized approximation algorithm for $O(1)$-wide TAP instances with approximation guarantee $2(\sqrt{3}-1) < 1.465$.
\end{lemma}
\begin{proof}
Given a $k$-wide TAP instance $(T,L)$ with $k=O(1)$, and a solution $(x,\lambda)$ to the $k$-wide-LP, from Corollary~\ref{cor:chvatal-rounding} we know that we can round $x$ to a solution $\mathcal{A}$ such that
\begin{equation}\label{eq:guaranteeOfA}
\frac{|\mathcal{A}|}{\OPT{}^*}\leq
2\ain-\aup+\across = 2-\across-\aup\enspace,
\end{equation}
where the equality follows from $\across+\ain =1$.
Furthermore, from Lemma \ref{lem:wired-rounding} we know that we can randomly round $(x,\lambda)$ to a solution 
$\mathcal{B}$ such that 
\begin{equation}\label{eq:solBBound}
\E[|\mathcal{B}|] \leq x(\Lin)+2x(\Lcross)-\E_{A}[\eta (A,\Lcr\cap
\begin{psmallmatrix}
A\\
2
\end{psmallmatrix}
)]\enspace,
\end{equation}
where $A$ fulfills the conditions of Lemma~\ref{lem:wired-rounding}.
We proceed to upper bound $\E[|\mathcal{B}|]/\OPT{}^*$. Let $K$ be the total number of leaves. Notice that $|\Vcr|< 2K$; for otherwise one would obtain that the average degree of $T$ is at least $2$, which contradicts that any tree has average degree strictly below $2$.
For each leaf node $v$ one must have $\sum_{\ell\in L:v\in \ell}x_{\ell}\geq 1$ for the edge incident to $v$ to be covered by the LP solution. Each link can contribute at most twice to the previous sums, and links in $\Lup\cup \Lncr$ at most once. Thus we can conclude that
\begin{equation}\label{eq:leaves-vs-x}
2x(\Lcross)+2x(\Lin)-x(\Lup)-x(\Lncr)
= 2\OPT^* - x(\Lup) - x(\Lncr) \geq K > |\Vcr|/2\enspace.
\end{equation}
Therefore, by Lemma \ref{lem:matching}, the expected size of the matching satisfies
\begin{equation*}
\begin{aligned}
\E_{A}\left[\eta(A, \Lcr\cap
\begin{psmallmatrix}
A\\
2
\end{psmallmatrix}
)
\right] \geq \frac{x(\Lcr)^2}{|\Vcr|}\geq 
\frac{(\acr \OPT{}^*)^2}{(4-2\aup-2\ancr)\OPT{}^*}
 =\frac{
(\acr)^2}{(4-2\aup-2\ancr)}\OPT^*\enspace.
\end{aligned}
\end{equation*}
Combining the above with~\eqref{eq:solBBound}, we conclude that
\begin{equation}\label{eq:guaranteeOfB}
\frac{\E[|\mathcal{B}|]}{\OPT{}^*}\leq 
\ain+2\across-\frac{(\across)^2}{4-2\aup-2\ancr}=1+\across
-\frac{(\acr)^2}{4-2\aup-2\ancr}\enspace.
\end{equation}
Hence, by returning the better of the solutions $\mathcal{A}$ and $\mathcal{B}$, we can upper bound the expected approximation factor of the algorithm by the minimum of the right-hand sides of~\eqref{eq:guaranteeOfA} and~\eqref{eq:guaranteeOfB}, i.e.,
\begin{align*}
1+\min\left\{1-\across-\aup,\across-\frac{(\acr)^2}{4-2\aup-2\ancr}\right\}\enspace.
\end{align*}
To obtain our approximation guarantee we maximize this function over $\aup,\across,\acr,\ancr\in [0,1]$ subject to the constraints:  
\begin{enumerate*}[label=(\roman*)]
\item $\aup\geq \ancr$ (by Lemma \ref{lem:up}),
\item $\across=\acr+\ancr$ (by definition), and
\item $\aup + \across \leq 1$ (because $\ain + \across =1$ and $\aup\leq \ain$).
\end{enumerate*} In Appendix~\ref{sec:calculus} (see Lemma~\ref{lem:maxAlphas}), we show via basic calculations that 
this maximum, achieved
for $\aup=\ancr=0$ and $\acr=\across=4-2\sqrt{3}$, is $2\cdot(\sqrt{3}-1)< 1.465$.
\end{proof}

\appendix

\section{Omitted Proofs in the Determination of our Approximation Factor}\label{sec:calculus}

In this section, we show a technical statement, stated as Lemma~\ref{lem:maxAlphas} below, about the maximizer of a function used in Lemma~\ref{lem:apx}, which completes the proof of Lemma~\ref{lem:apx}.

\begin{lemma}\label{lem:maxAlphas}
The maximum of
\begin{equation*}
1+\min\left\{1-\across-\aup,\across-\frac{(\acr)^2}{4-2\aup-2\ancr}\right\}
\end{equation*}
over $\aup,\across,\acr,\ancr\in [0,1]$ subject to the constraints:  
\begin{enumerate}[label=(\roman*),nosep,topsep=0.3em]
\item $\aup\geq \ancr$, 
\item\label{item:acrossDecomp} $\across = \acr + \ancr$, and
\item\label{item:upCrossBound} $\aup + \across \leq 1$.
\end{enumerate}
is achieved for $\aup = \ancr = 0$ and $\acr=\across=4-2\sqrt{3}$, leading to a value of $2(\sqrt{3}-1)$.
\end{lemma}
\begin{proof}
The function is non-increasing in $\aup$, and we can hence set $\aup=\ancr$. Moreover, using $\acr = \across - \ancr = \across - \aup$, we
have to maximize \begin{align*}
f(\across,\aup) \coloneqq 1+\min\left\{1-\across-\aup,\across-\frac{(\across-\aup)^2}{4(1-\aup)}\right\}\enspace,
\end{align*}
subject to the constraints
\begin{enumerate}[label=(\alph*),nosep,topsep=0.3em]
\item $\aup\leq\across$ (because $\aup=\ancr\leq \across$ due to~\ref{item:acrossDecomp}), and
\item $\aup + \across \leq 1$ (due to~\ref{item:upCrossBound}).
\end{enumerate}

Let $f_1(x,y)=1-x-y$, $f_2(x,y)=x-\frac{(x-y)^2}{4(1-y)}$ and $f(x,y)=\min\{f_1(x,y),f_2(x,y)\}$. Using this notation, the maximization problem we consider is 
\begin{equation}\label{eq:maxProbSimpl}
\max_{0\leq y\leq x\leq 1-y} f(x,y)\enspace.
\end{equation}

Let $D=\{(x,y): 0\leq y\leq x\leq 1-y\}$ be the domain over which we maximize, and let $(x^*,y^*)\in D$ be a maximizer of~\eqref{eq:maxProbSimpl}.
We first show that any such maximizer $(x^*, y^*)$ must satisfy $f_1(x^*,y^*)=f_2(x^*,y^*)$.

Indeed, if $f_1(x^*,y^*)<f_2(x^*,y^*)$, then $(x^*,y^*)$ is a maximizer of $\max_{(x,y)\in D} \max f_1(x,y)$. However, the only maximizer
of $f_1$ over $D$ is clearly $x=y=0$, for which we have $f_1(0,0) > f_2(0,0)$, which contradicts $f_1(x^*,y^*)< f_2(x^*,y^*)$.

On the other hand, if $f_1(x^*,y^*)> f_2(x^*,y^*)$, then $(x^*,y^*)$ is a 
maximizer of $\max_{(x,y)\in D} f_2(x,y)$. Observe that $\frac{\partial}{\partial x} f_2(x,y)=1-\frac{2x-2y}{4(1-y)}$ 
is nonnegative for all $(x,y)\in D$, which
follows directly from the fact that $(x,y)\in D$ implies $y\leq \frac{1}{2}$.
Therefore, if $(x^*,y^*)$ is a maximizer of $f_2$ over $D$, it holds that $y^*=1-x^*$. Hence, 
$f_2(x^*,y^*)=x^*-\sfrac{(2x^*-1)^2}{4x^*}=1-\sfrac{1}{4x^*}$, which implies that $(x^*,y^*)=(1,0)$, which is the unique maximizer 
of $f_2(x,y)$ over all points in $D$ satisfying $y=1-x$. However, this contradicts that $f_1(x^*,y^*)>f_2(x^*,y^*)$.

Hence, any maximizer $(x^*,y^*)$ of $f(x,y)$ over $D$ satisfies $f_1(x^*,y^*)=f_2(x^*,y^*)$, as 
desired. We thus only need to consider points $(x,y)\in D$ satisfying $f_1(x,y)=f_2(x,y)$, when seeking a maximizer 
of~\eqref{eq:maxProbSimpl}.

By expanding $f_1(x,y)=f_2(x,y)$, we get
\begin{equation*}
-x^2 - 6xy + 8x - 5y^2 +8y -4 = 0 \enspace.
\end{equation*}
Solving the above equation for $x$ gives
\begin{equation*}
x \in \left\{
4 - 3y - 2 \sqrt{(3-y)(1-y)}\;,\;
4 - 3y + 2 \sqrt{(3-y)(1-y)}
\right\}\enspace.
\end{equation*}
Because $x\leq 1$ and $y\leq \sfrac{1}{2}$, we must have
\begin{equation}\label{eq:xFromY}
x = 4 - 3y - 2 \sqrt{(3-y)(1-y)}\enspace.
\end{equation}
By substituting $x$ using the above equation, we obtain
\begin{equation*}
f_1(x,y)=f_2(x,y)=
2 y + 2 \sqrt{\left(3 - y\right) \left(1-y\right)} - 3\enspace \eqqcolon h(y)\enspace.
\end{equation*}
Finally, observing that the derivative of $h$ with respect to $y$, which is given by 
\begin{equation*}
\frac{dh(y)}{dy} = \frac{2 y - 4}{\sqrt{(3-y)(1-y)}} + 2\enspace,
\end{equation*}
is non-positive for $y\in [0,\sfrac{1}{2}]$, we have that $h(y)$ is maximized for $y^*=0$. By~\eqref{eq:xFromY}, this leads to $x^*=4-2\sqrt{3}$. Because these values for $x^*$ and $y^*$ satisfy $(x^*,y^*)\in D$, they are indeed a maximizer for~\eqref{eq:maxProbSimpl}, with resulting value $f_1(x^*,y^*) = f_2(x^*,y^*) = 2(\sqrt{3}-1)$, as desired.
\end{proof}

\section{Reduction to $k$-Wide Instances}\label{sec:reductionToWideTrees}

In this section, we show Theorem~\ref{thm:decomp}, i.e., that general TAP instances can be reduced to $O(1)$-wide ones by losing only a small constant in the approximation guarantee.
Let $k\in \mathbb{Z}_{\geq 1}$, and let $\mathcal{A}$ be an algorithm that is an $\alpha$-approximation for TAP on $k$-wide instances.

A key ingredient in our reduction is the classical cut-LP for TAP, which, we recall, is given by 
\begin{equation*}
\min\left\{x(L)  \;\middle\vert\; x\in \Pi_T\right\}\enspace.
\end{equation*}
To reduce general TAP to $O(1)$-wide TAP, we start with the above cut-LP, and then use the ellipsoid method to strengthen the LP through cutting planes that we generate by using the algorithm $\mathcal{A}$ on well-chosen sub-problems that are $O(1)$-wide.
To better illustrate the idea of our approach, we first state a reduction result that can be derived by using techniques introduced by Adjiashvili~\cite{adjiashvili2017beating}, and which we crucially exploit later on. In short, the result shows that the problem of finding a good rounding algorithm for the cut-LP can be reduced to finding a good rounding algorithm for $O(1)$-wide instances.

In the lemma below, a partition of $T=(V,E)$ into subtrees $T_1=(V_1,E_1),\ldots,T_q=(V_q,E_q)$ means that $T_1,\ldots, T_q$ are connected subgraphs of $T$ (hence trees), and $\{E_1,\ldots, E_q\}$ partitions $E$. Moreover, the notation $T_i/H_i$, where $H_i\subseteq E_i$ is a subset of the edges of $T_i$, denotes the graph obtained from $G_i$ by contracting $H_i$.
\begin{lemma}\label{lem:decompALaDavid}
Let $k\in \mathbb{Z}_{\geq 1}$, let $(T,L)$ be a TAP instance, and let $x\in \Pi_T$.
Then one can efficiently partition $T$ into subtrees $T_i=(V_i,E_i)$ for $i\in [q]$ and find $H_i\subseteq E_i$ for $i\in [q]$ such that:
\begin{enumerate}[label=(\roman*),itemsep=-0.2em,topsep=0.3em]
\item\label{item:thinSubInstances} For $i\in [q]$, $T_i/H_i$ is $k$-wide.
\item\label{item:canCoverHi} One can efficiently get a link set $M\subseteq L$ with $|M| = O(\sfrac{1}{\sqrt{k}}) \cdot x(L)$ that covers all edges of $\cup_{i=1}^q H_i$.
\item\label{item:goodLpSolForSubProbs} $\sum_{i=1}^q x(\cov(E_i)) =  (1+O(\sfrac{1}{\sqrt{k}}))\cdot x(L)$.
\end{enumerate}
Moreover, one can choose the edge sets $H_i$ for $i\in [q]$ to form subtrees, and in the $k$-wide trees $T_i/H_i$, the contracted node that corresponds to $H_i$ can be chosen as the root.
\end{lemma}

The above lemma follows quite directly from techniques introduced by Adjiashvili~\cite{adjiashvili2017beating}, with some smaller changes.
Still, for completeness, we provide a formal proof of it in Appendix~\ref{sec:decomp}.
In what follows, we do not strictly need the 
fact that the $H_i$ can be chosen to form trees. Still, we want to highlight this additional property, which can lead 
to considerable simplifications when designing reductions to wide instances that do not rely on the ellipsoid method.
In particular, one could use our Lemma~\ref{lem:decompALaDavid} in Adjiashvili's approach~\cite{adjiashvili2017beating}, and avoid most technicalities involving so-called compound nodes.

Notice that Lemma~\ref{lem:decompALaDavid} indeed implies that a solution $x\in \Pi_T$ to the cut-LP for a TAP instance 
can be broken down into solutions to the cut-LPs of independent $k$-wide TAP instances, by increasing the objective 
value by a factor $1+O(\sfrac{1}{\sqrt{k}})$.
To see this, consider the problem of finding a smallest set of links that covers only the edges $E_i$ for some $i\in 
[q]$, where $T_j=(V_j,E_j)$ for $j\in[q]$ are the subtrees as stated in Lemma~\ref{lem:decompALaDavid}. This problem can 
be interpreted as a TAP instance on the subtree $T_i=(V_i,E_i)$, by interpreting each link $\ell\in L$ as a link that 
covers $P_\ell\cap E_i$.
Hence, the vector $x^i$ obtained from $x$ by setting to $0$ all coordinates corresponding to links in $L\setminus 
\cov(E_i)$, is a cut-LP solution for the TAP instance on $T_i$.
Moreover, the total cost of these cut-LP solutions is
\begin{equation*}
\sum_{i=1}^q x^i(L) = \sum_{i=1}^q x(\cov(E_i)) = (1+O(\sfrac{1}{\sqrt{k}}))\cdot x(L)\enspace,
\end{equation*}
where the last equality follows by point~\ref{item:goodLpSolForSubProbs} of Lemma~\ref{lem:decompALaDavid}, and is 
hence only by a factor $1+O(\sfrac{1}{\sqrt{k}})$ larger than the LP cost of $x$. Finally, by 
Lemma~\ref{lem:decompALaDavid}~\ref{item:canCoverHi} we can cover all edges in $\cup_{i=1}^q H_i$ at small cost. Thus, 
for each $i\in [q]$, it remains to cover the edges $E_i\setminus H_i$, which, after contracting $H_i$ which we already 
covered, reduces to a $k$-wide TAP instance on $T_i/H_i$ by point~\ref{item:thinSubInstances} of 
Lemma~\ref{lem:decompALaDavid}.

The above outline, and additional details provided in the proof of Lemma~\ref{lem:decompALaDavid}, highlights how one 
can try to reduce LP-based approaches to wide instances.
In particular, it shows that if $O(1)$-wide TAP instances have small integrality gap with respect to the cut-LP, then the cut-LP has a small integrality gap for any instance.
A similar reduction approach was key in some recent progress on TAP, and also WTAP with a bounded ratio between maximal and minimal cost~\cite{adjiashvili2017beating,fiorini_2018_approximating}. Even though such reductions have a certain versatility, so far, they only have been carried out for specific linear programs. For example, Adjiashvili~\cite{adjiashvili2017beating} extended the cut-LP with ``bundle''-constraints, and Fiorini et~al.~\cite{fiorini_2018_approximating} additionally added $\{0,\sfrac{1}{2}\}$-Chv\'atal-Gomory-constraints. In both cases the authors had to show that key properties are preserved by the decomposition.
In the following, we show how, based on the decomposition guaranteed by Lemma~\ref{lem:decompALaDavid}, one can reduce TAP instances to $O(1)$-wide ones without assuming a particular LP-based approach.

To obtain a black-box reduction, we show that the following holds. Given an $\alpha$-approximation algorithm for $k$-wide TAP instances, one can efficiently compute a point $x\in \Pi_T$ and an $\alpha$-approximation for each $k$-wide tree in the decomposition obtained by Lemma~\ref{lem:decompALaDavid}. This result is formalized in the following statement.
\begin{lemma}\label{lem:boostedDecomp}
Let $k\in \mathbb{Z}_{\geq 1}$.
Given an $\alpha$-approximation algorithm $\mathcal{A}$ for TAP on $k$-wide instances, one can, for any TAP instance $(T,L)$, efficiently compute:
\begin{enumerate}[label=(\roman*),nosep]
\item\label{item:xBoundedByOpt} A point $x\in \Pi_T$ satisfying $x(L)\leq \nu^*$, where $\nu^*$ is the optimal value of $(T,L)$,
\item a partition $T_i=(V_i, E_i)$ with $H_i\subseteq E_i$  $\forall i\in [q]$ of $T$ w.r.t.~$x$ as described by Lemma~\ref{lem:decompALaDavid}, and
\item\label{item:LiCheap} for $i\in [q]$, a link set $L_i\subseteq L$ covering $E_i\setminus H_i$ and satisfying $|L_i| \leq \alpha\cdot x(\cov(E_i\setminus H_i))$.
\end{enumerate}
\end{lemma}
We first observe that Theorem~\ref{thm:decomp} is an immediate consequence of the above lemma.
\begin{proof}[Proof of Theorem~\ref{thm:decomp}]
Because the partition $T_i=(V_i,E_i)$ and $H_i\subseteq E_i$ for $i\in [q]$ of $T$ satisfy the conditions of 
Lemma~\ref{lem:decompALaDavid}, we can efficiently obtain a set of links $M\subseteq L$ by 
Lemma~\ref{lem:decompALaDavid} such that $|M| = O(\sfrac{1}{\sqrt{k}})\cdot x(L)$ and $M$ covers all edges of $\cup_{i=1}^q 
H_i$. We return the solution $Q\coloneqq M \cup \bigcup_{i=1}^q L_i$. $Q$ is clearly a solution to the TAP instance 
$(T,L)$. It remains to show that its approximation guarantee is $\alpha+O(\sfrac{1}{\sqrt{k}})$, which follows due to
\begin{align*}
\left\vert M \cup \bigcup_{i=1}^q L_i\right\vert
  &\leq  |M| +\sum_{i=1}^q |L_i| \\
  &\leq O\left(\frac{1}{\sqrt{k}}\right)\cdot x(L) + \alpha\cdot \sum_{i=1}^q x(\cov(E_i\setminus H_i))\\
  &\leq O\left(\frac{1}{\sqrt{k}}\right)\cdot x(L) + \alpha\cdot \left(1+O\left(\frac{1}{\sqrt{k}}\right)\right) \cdot x(L)\\
  &= \left(\alpha + O\left(\frac{1}{\sqrt{k}}\right)\right)\cdot x(L)\\
  &= \left(\alpha + O\left(\frac{1}{\sqrt{k}}\right)\right)\cdot \nu^*\enspace,
\end{align*}
where the second inequality follows from Lemma~\ref{lem:decompALaDavid} \ref{item:canCoverHi} and 
Lemma~\ref{lem:boostedDecomp}~\ref{item:LiCheap}, the third inequality follows from 
Lemma~\ref{lem:decompALaDavid}~\ref{item:goodLpSolForSubProbs}, and the last equation from Lemma~\ref{lem:boostedDecomp}~\ref{item:xBoundedByOpt}.
\end{proof}

It remains to show Lemma~\ref{lem:boostedDecomp}. To this end, let $(T,L)$ be a TAP instance on a tree $T=(V,E)$ with links $L\subseteq \binom{V}{2}$.
Consider the following polytope, parameterized by a value $\nu\in \mathbb{Z}_{\geq 0}$ and $k\in 
\mathbb{Z}_{\geq 1}$, and where we denote by $\OPT(S)$, for any edge set $S\subseteq E$, the minimum number of links 
needed to cover all edges of $S$.
\begin{equation*}
\Pi_T(\nu,k) \coloneqq \left\{
x\in \Pi_T \;\middle\vert\!
\begin{array}{r@{\;\,}c@{\;\,}l@{\quad}l}
x(\cov(U\setminus H))  &\geq  &\OPT(U\setminus H)
    & \forall U\subseteq E, H\subseteq U \text{ s.t.~}(V,U)/H \text{ is $k$-wide,}\\
x(L) &\leq &\nu\enspace &
\end{array}
\!\right\}
\end{equation*}
To clarify, stating that $(V,U)/H$ is $k$-wide in the first set of constraints, implies that $U\subseteq E$ and $H\subseteq U$ must be such that $(V,U)/H$ is a tree (that is $k$-wide).

In words, $\Pi_T(\nu,k)$ describes all cut-LP solutions of value at most $\nu$ with the additional requirement that on any 
$k$-wide sub-instance $(V,U)/H$ defined by edge sets $H\subseteq U \subseteq E$, the $x$-value of the links covering 
edges in $(V,U)/H$ is no lower than the optimal value for this sub-instance.
Notice that if $\nu^*\in \mathbb{Z}_{\geq 0}$ is the optimal value of the original instance, then $\Pi_T(\nu^*,k)$ 
contains the characteristic vector of an optimal solution, and is thus indeed a relaxation of the considered TAP 
problem. Note that we can easily ``guess'' the optimal value $\nu^*$, because it lies within $\{0,\ldots, |L|\}$, and 
thus assume that $\nu^*$ is known.\footnote{As we will discuss later, we can also use a binary search technique to get 
an approximation of $\nu^*$ that is good enough for us. This speeds up the procedure, which is interesting 
when extending our reduction to weighted instances with small $\sfrac{c_{\max}}{c_{\min}}$.}

Using Lemma~\ref{lem:decompALaDavid}, it is not hard to see that the integrality gap of the linear program $\min\{x(L) 
\mid x\in \Pi_T(\nu^*,k)\}$ is bounded by $1+O(\sfrac{1}{\sqrt{k}})$.
\footnote{Indeed, we can apply Lemma~\ref{lem:decompALaDavid} to an optimal solution $y$ to $\min\{x(L) \mid x\in \Pi_T(\nu^*,k)\}$, because $\Pi_T(\nu^*,k) \subseteq \Pi_T$. This allows for partitioning $T$ into subtrees $T_i=(V_i,E_i)$ for $i\in [q]$ with $H_i\subseteq E_i$ that satisfy the conditions of Lemma~\ref{lem:decompALaDavid} with respect to $y$.
Point~\ref{item:canCoverHi} of the lemma shows that covering all edges $\cup_{i=1}^q H_i$ can be done with at most $O(\sfrac{1}{\sqrt{k}})\cdot y(L)$ links. Moreover, Lemma~\ref{lem:decompALaDavid}~\ref{item:goodLpSolForSubProbs} shows that if, for each $i\in [q]$, the edge set $E_i\setminus H_i$ can be covered by at most $y(\cov(E_i))$ links, then the obtained solution uses at most $(1+O(\sfrac{1}{\sqrt{k}})) y(L)$ links, thus leading to an integrality gap of $1+O(\sfrac{1}{\sqrt{k}})$. This is indeed possible because $\Pi_T(\nu^*,k)$ contains constraints that guarantee that $E_i\setminus H_i$ can be covered by at most $y(\cov(E_i\setminus H_i))\leq y(\cov(E_i))$ links, as long as $(V,E_i)\setminus H_i$ is $k$-wide, which holds due to Lemma~\ref{lem:decompALaDavid}~\ref{item:thinSubInstances}.
}
However, optimizing (or separating) over $\Pi_T(\nu^*,k)$ is 
difficult. We therefore use a \emph{partial separation oracle} that correctly never separates a point $y\in \mathbb{R}^L$ that is in $\Pi_T(\nu^*,k)$, but may sometimes not separate points $y\in \mathbb{R}^L$ that are not in $\Pi_T(\nu^*,k)$.
The partial separation oracle for a point $y\in \mathbb{R}^L$, and an arbitrary value $\nu\in \mathbb{Z}_{\geq 0}$, is 
described by Algorithm~\ref{alg:partialSepOracle}.

For simplicity, our description is given for $\mathcal{A}$ being a deterministic $\alpha$-approximation algorithm for 
$k$-wide instances.
The extension to randomized algorithms $\mathcal{A}$ can be done with standard techniques. More precisely, if $\mathcal{A}$ is a randomized procedure, one can, for any constant $\delta >0$, obtain a $(\alpha+\delta)$-approximation with high probability in polynomial time, by doing multiple runs of $\mathcal{A}$ on the same instance and returning the best outcome.\footnote{Notice that a usual application of Markov's inequality would guarantee that we obtain a $(1+O(\delta))\cdot \alpha$-approximation with high probability. However, because we can assume $\alpha\leq 2$, as $2$-approximations are well-known even for WTAP, this implies that we can obtain an $(\alpha+\delta)$-approximation.}
Hence, because we focus on polynomial-time procedures, we can assume that with high probability, $\mathcal{A}$ returns a $(\alpha+\delta)$-approximation whenever it is called in the following.

\begin{algorithm2e}[h]
\begin{enumerate}[leftmargin=0cm,rightmargin=1cm, itemsep=0.2em]
\item If $y(L) > \nu$, return the separating hyperplane $x(L)\leq \nu$.

\item If $y\not\in \Pi_T$, return a separating hyperplane that separates $y$ from $\Pi_T$. 

\item\label{item:sepOverKWide} Decompose $T$ into $T_i  = (V_i,E_i)$ for $i\in [q]$ and find sets $H_i\subseteq E_i$ for $i\in [q]$ by using Lemma~\ref{lem:decompALaDavid} with $x=y$ and parameter $k$.
For each $i\in [q]$, use $\mathcal{A}$ to obtain an $\alpha$-approximation to the $k$-wide TAP instance $T_i/H_i$; let $L_i\subseteq L$ be the obtained solution.
If $y(\cov(E_i\setminus H_i))< \lceil \sfrac{|L_i|}{\alpha}\rceil$ for some $i\in [q]$, then return the separating hyperplane $x(\cov(E_i\setminus H_i)) \geq \lceil \sfrac{|L_i|}{\alpha}\rceil$.

\end{enumerate}
\caption{Partial separation oracle for $\Pi_T(\nu,k)$}
\label{alg:partialSepOracle}
\end{algorithm2e}
Notice that whenever Algorithm~\ref{alg:partialSepOracle} returns a separating hyperplane for some point $y\in 
\mathbb{R}^L$, then this is indeed a valid one. This is clear for the first two steps of the algorithm. Moreover, if 
step~\ref{item:sepOverKWide} returns a hyperplane $x(\cov(E_i\setminus H_i)) \geq \lceil \sfrac{|L_i|}{\alpha}\rceil$, 
then this hyperplane indeed separates $y$ over $\Pi_T(\nu,k)$ because $|L_i| \leq \alpha \cdot \OPT(E_i\setminus H_i)$, 
as $L_i$ is an $\alpha$-approximate solution for the TAP problem on $T_i$, and $\OPT(E_i\setminus H_i)$ is integer. 
Thus, $\lceil \sfrac{|L_i|}{\alpha}\rceil \leq \OPT(E_i\setminus H_i)$.

We now show that a point $x\in \Pi_T$, and links $L_i\subseteq L$ for $i\in [q]$ as guaranteed by 
Lemma~\ref{lem:boostedDecomp}, can be obtained efficiently by using any version of the ellipsoid method that efficiently 
returns a point in $\Pi_T(\nu^*,k)$; however we run the ellipsoid method with the partial separation oracle given by 
Algorithm~\ref{alg:partialSepOracle} instead of a true separation oracle.
\begin{proof}[Proof of Lemma~\ref{lem:boostedDecomp}]
As mentioned, we could guess the optimal value $\nu^*$ of the TAP problem, because $\nu^*\in \{0,\ldots, |L|\}$. However, we will instead use a binary search technique that only leads to an additional multiplicative factor of $\log|L|$ in the running time (instead of $|L|$). This is useful in our discussion of how our reduction can be extended to some weighted instances.

Consider a fixed $\nu\in \{0,\ldots, |L|\}$. We run a classical ellipsoid type method to determine a point in 
$\Pi_T(\nu,k)$, if $\Pi_T(\nu,k)$ is non-empty. If we had a true (instead of a partial) separation oracle for $\Pi_T(\nu,k)$, 
then this could be done with a polynomial number of oracle calls and further operations taking polynomial time in the 
encoding lengths of the constraints. Notice that the constraints defining $\Pi_T(\nu,k)$ have small encoding length 
because all left-hand sides are $0/1$-vectors, and the right-hand sides are integers within $\{0,\ldots, |L|\}$. (We 
refer the interested reader to the excellent book of Gr\"otschel, Lov\'asz, and 
Schrijver~\cite{groetschel_1993_geometric} for details on the ellipsoid method.)
This implies that, assuming $\Pi_T(\nu,k)\neq \emptyset$, if we run the ellipsoid method to find a point in $\Pi_T(\nu,k)$ 
with the partial separation oracle, then one of the following happens:
\begin{enumerate}[label=(\roman*),nosep,topsep=0.4em]
\item\label{item:alwaysSeparate} The partial separation oracle will throughout the ellipsoid method always return a 
separating hyperplane whenever it is called for a point $y\not\in\Pi_T(\nu,k)$. In this case, we get a point $x\in \Pi_T(\nu,k)$ in polynomial time as desired.

\item\label{item:onceNotSeparated} At some point during the ellipsoid method, the separation oracle will be invoked with a point $y\in \mathbb{R}^L$ for which no separating hyperplane is found.
\end{enumerate}
By employing binary search on $\nu\in \{0,\ldots, |L|\}$, we find a value $\bar{\nu}\in \{0,\ldots, |L|\}$ such that running 
the above-described ellipsoid method with partial separation oracle on $\Pi_T(\bar{\nu}-1,k)$ returns that 
$\Pi_T(\bar{\nu}-1,k)=\emptyset$; however, for $\Pi_T(\bar{\nu},k)$ it returns that the polytope is not empty. Such a 
$\bar{\nu}$ can clearly be found through binary search over $\{0,\ldots, |L|\}$ by $O(\log|L|)$ calls of the ellipsoid 
method with partial separation oracle. Also notice that $\bar{\nu} \leq \nu^*$, because the partial separation oracle 
is ``weaker'' than a true separation oracle, i.e., whenever it separates a point $y\in \mathbb{R}^L$, then $y\not\in 
\Pi_T(\nu,k)$; however, sometimes points not in $\Pi_T(\nu,k)$ do not get separated.

Consider the run of the ellipsoid method with partial separation oracle on $\Pi_T(\bar{\nu},k)$. If~\ref{item:alwaysSeparate} happens, then we get a point $x\in \Pi_T(\bar{\nu},k)$ and a corresponding partition of $T$ into $T_i=(V_i,E_i)$ for $i\in [q]$ and $H_i\subseteq E_i$ for $i\in [q]$, as stated in Lemma~\ref{lem:decompALaDavid}. Running for $i\in [q]$ the $\alpha$-approximation $\mathcal{A}$ on the $k$-wide instance $T[E_i]/H_i$, a link set $L_i\subseteq L$ is obtained such that
\begin{equation*}
|L_i| \leq \alpha\cdot \OPT(E_i\setminus H_i) \leq \alpha \cdot x(\cov(E_i\setminus H_i))\enspace,
\end{equation*}
where the second inequality follows from $x\in \Pi_T(\bar{\nu},k)$. Hence, the point $x$ together with the partition of $T$, and the links $L_i$ for $i\in [q]$ fulfill the conditions of Lemma~\ref{lem:boostedDecomp} as desired.

Now assume that~\ref{item:onceNotSeparated} occurs, and let $x\in \mathbb{R}^L$ be a point for which the ellipsoid 
method called the partial separation oracle, and no separating hyperplane was returned for $x$. The partial separation 
oracle computed a partition of $T$ into $T_i=(V_i,E_i)$ with sets $H_i\subseteq E_i$ for $i\in [q]$ using 
Lemma~\ref{lem:decompALaDavid}. Moreover, for each $i\in [q]$, a link set $L_i\subseteq L$ was determined that satisfies 
$x(\cov(E_i\setminus H_i))\geq \lceil \sfrac{|L_i|}{\alpha}\rceil$, because $x$ was not separated by 
step~\ref{item:sepOverKWide} of the partial separation procedure. Hence, for each $i\in [q]$, we have $|L_i| \leq 
\alpha\cdot x(\cov(E_i\setminus H_i))$.
Again, the point $x$ together with the partition of $T$ and the links $L_i$ for $i\in [q]$ satisfy the conditions of Lemma~\ref{lem:boostedDecomp}.
\end{proof}

\section{Decomposing Into $k$-Wide Trees}\label{sec:decomp}

In this section, we show Lemma~\ref{lem:decompALaDavid}. As already mentioned, this statement follows quite directly from techniques introduced by Adjiashvili~\cite{adjiashvili2017beating}, with some smaller changes.
Still, for completeness, we provide a formal proof of the lemma in this section, by closely following Adjiashvili's approach. 

We are given a TAP instance $(T,L)$ and a vector $x\in \Pi_T$, and our goal is to construct a decomposition into $q$ subtrees $D=\{T_1,\ldots, T_q\}$ fulfilling the conditions of Lemma~\ref{lem:decompALaDavid}.
As in~\cite{adjiashvili2017beating}, to construct $D$, we rely on the concept of $\gamma$-light edges, which we introduce next. 
For a given TAP instance $(T=(V,E),L)$ and edge $e=\{u,v\}\in E$, let $T(e,u)$ be the subtree of $T$ that is 
obtained by removing $e$ and considering the obtained connected component containing $u$. 
Let $E_T(e,u)$ be the corresponding set of edges.
We proceed to define $\gamma$-light edges, which are edges on which the instance can be split by only incurring a small increase in cost. In words, $\gamma$-light edges with respect to some point $x\in \Pi_x$ are edges $e$ such that the $x$-value of links covering $e$ is small compared to the $x$-value of links on each of the two connected components of the tree when removing $e$. Because we need this notion also for subtrees, we define it with respect to any subtree $T_i$ of $T$.
\begin{definition}\label{def:light}
Let $(T,L)$ be a TAP instance, let $x\in \Pi_T$, and let $\gamma >0$.
Given a subtree $T_i$ of $T$, an edge $e=\{u,v\}$ of $T_i$ is called \emph{$\gamma$-light} in $T_i$ if 
\begin{equation}\label{eq:gammaLight}
x(\cov(e))\leq \gamma \cdot \min\left\{x(\cov(E_{T_i}(e,u))\setminus\cov(e)),x(\cov(E_{T_i}(e,v))\setminus
\cov(e))\right\}\enspace.
\end{equation}
\end{definition}

Analogous to Adjiashvili's approach, the decomposition $D$ will be constructed by applying a splitting operation iteratively. We originally set $D$ to $\{T\}$. 
First, we define what a splitting operation consists of:
\begin{definition}
	Given a $\gamma$-light edge $e=\{u,v\}$ in any $T_i\in D$, a \emph{split} consists of the following steps:
	\begin{enumerate}[label=(\roman*),nosep,itemsep=0.2em,topsep=0.2em]
		\item Remove $T_i$ from $D$.
		\item Let $T^u$ be the maximal subtree of $T_i$ that contains vertex $v$, and none of $v$'s neighbors 
except for $u$. Conversely, let $T^v$ be the maximal subtree of $T_i$ that contains $v$ and all of $v$'s neighbors, 
except for $u$. Insert $T^u$ and $T^v$ into $D$.
	\end{enumerate}
\end{definition}

Next, we have to specify \emph{which} light edge we will choose to split at any given iteration. To simplify our 
analysis, 
we want to choose light edges that produce one subtree with no light edges. We have the following observation:
\begin{observation}
Given a TAP instance $(T,L)$ and $x\in \Pi_T$,
if we can split some subtree $T'$ of $T$ into $T^1$ and $T^2$ and some edge $e$ is $\gamma$-light in $T^1$, then $e$ is 
light in $T'$ as well.
\end{observation}
This observation immediately follows from the fact that for any subtree $T^1$ of $T$ and any edge $e=\{u,v\}$ that is part of $T^1$, we have $E_{T^1}(e,u)\subseteq E_T(e,u)$.
The above observation produces the following corollary, regarding the existence of good splits:
\begin{corollary}
Given a TAP instance $(T,L)$ and a vector $x\in \Pi_T$, if there exists a $\gamma$-light edge in subtree $Q$ of $T$ then 
there exists a $\gamma$-light edge $\{u,v\}$ in $Q$ such that splitting $Q$ into $Q^u$ and $Q^v$ implies that at least one of $Q^u$ and $Q^v$ contain no $\gamma$-light edges.
\end{corollary}

We call a a subtree with no light edges an \emph{unsplittable component} (or simply \emph{unsplittable}), and an edge $\{u,v\}$ as described in the above corollary is called \emph{$Q$-critical}.

Remember that  originally $D=\{T\}$. The decomposition process to construct $D$ is now straightforward: while there 
exists a $\gamma$-light edge in some $T_i\in D$, choose a $T_i$-critical $\gamma$-light edge $e$, and split $T_i$ at 
$e$. 

For convenience, we will index the subtrees in $D$ in the order by which they were introduced in $D$, with $T_1$ being the first one introduced and $T_q$ the last one. (To be precise, the last splitting step created two subtrees that are contained in $D$; it does not matter which one of the two is the last one in the numbering and which one the second-to-last one.)
For $1\leq i \leq q-1$, let $e_i$ be the edge that was split in order to introduce $T_i$ into $D$.
We now show that $D=\{T_1,\ldots, T_q\}$ obtained by successive splitting at $\gamma$-light edges, where $\gamma = \frac{1}{\sqrt{k}}$ with $k$ being the parameter in Lemma~\ref{lem:decompALaDavid}, satisfies the conditions of of Lemma~\ref{lem:decompALaDavid}, for well-chosen $H_i\subseteq E_i$ for $i\in [q]$.
\begin{proof}[Proof of Lemma~\ref{lem:decompALaDavid}]
 We start by proving that the trees $T_i=(V_i,E_i)$ in the decomposition $D$ satisfy point~\ref{item:goodLpSolForSubProbs} of Lemma~\ref{lem:decompALaDavid}. Because the trees $T_i$ for $i\in [q]$ do not share any edges, we have
\begin{equation}
\sum_{i=1}^q x\left(L\cap
\begin{psmallmatrix}
V_i\\
2
\end{psmallmatrix}
\right)\leq x(L)\enspace. \label{eqn:leq-x(L)} 
\end{equation}
Furthermore, observe that
\begin{equation}
\sum_{i=1}^q x(\cov(E_i))=\sum\limits_{i=1}^q
x\left(L\cap
\begin{psmallmatrix}
V_i\\
2
\end{psmallmatrix}
\right)
+2\sum\limits_{i=1}^{q-1}x(\cov(e_i))\enspace.\label{eqn:2*cov} 
\end{equation}
To see this, consider the \emph{split-tree} $S$, that contains a vertex $s_{T_i}$ for each subtree $T_i\in D$, and an 
edge $\{s_{T_i},s_{T_j}\}$ if $T_i$ and $T_j$ share a vertex (remember that different subtrees in $D$ will share some 
vertices, but no edges). Now, root $S$ at vertex $s_{T_q}$, and turn $S$ into a directed graph by directing all the 
edges away from the root. The key observation is that, for any $T_i\in D\setminus \{T_q\}$, 
\begin{equation*}
x(\cov(E_i))\leq x\left(L\cap
\begin{psmallmatrix}
V_i\\
2
\end{psmallmatrix}
\right)+x(\cov(e_i))+\sum\limits_{(s_{T_i},s_{T_j})\in \delta^+_S(s_{T_i})} 
x(\cov(e_j))\enspace,
\end{equation*}
while for $T_q$ we have
$$
x(\cov(E_q))\leq x(L\cap
\begin{psmallmatrix}
V^q\\
2
\end{psmallmatrix}
)+\sum\limits_{(s_{T_q},s_{T_j})\in \delta^+_S(s_{T_q})} x(\cov(e_j))\enspace.
$$
Summing up $x(\cov(E_i))$ over all $i\in [q]$, we get (\ref{eqn:2*cov}). Now, from  (\ref{eqn:leq-x(L)}) and 
(\ref{eqn:2*cov}), we conclude that in order to prove~\ref{item:goodLpSolForSubProbs}, it suffices to show
$$
 \sum\limits_{i=1}^{q-1}x\left(\cov(e_i)\right) = O\left(\frac{1}{\sqrt{k}}\right)\cdot x(L)\enspace.
$$
To do so, observe that, for any $i\in[q-1]$, due to the decomposition process we have
\begin{equation}\label{eqn:gamma-cov}
x(\cov(e_i))\leq \gamma \cdot x(\cov(E_i)\setminus\cov(e_i))\enspace,
\end{equation}
and moreover
\begin{equation}\label{eqn:sum-ei-leq-x}
\sum\limits_{i=1}^{q-1} x(\cov(E_i)\setminus\cov(e_i))\leq x(L)\enspace,
\end{equation}
because for any $i,j\in [q]$ with $i\neq j$, we have $(\cov(E_i)\setminus \cov(e_i))\cap (\cov(E_j)\setminus \cov(e_j))=\emptyset$. Indeed, any link $\ell \in \cov(E_i)\cap \cov(E_j)$ must be in either $\cov(e_i)$ or $\cov(e_j)$, because by removing $e_i$ and $e_j$ from $T$, the two edge sets $E_i$ and $E_j$ are in different connected components.
Finally, from~\eqref{eqn:gamma-cov} and~\eqref{eqn:sum-ei-leq-x} we obtain point~\ref{item:goodLpSolForSubProbs} of Lemma~\ref{lem:decompALaDavid}:
\begin{equation}
\sum\limits_{i=1}^{q-1} x(\cov(e_i))\leq \gamma \cdot x(L) = \frac{1}{\sqrt{k}}\cdot x(L)\enspace.
\end{equation}

To show points~\ref{item:thinSubInstances} and~\ref{item:canCoverHi} of Lemma~\ref{lem:decompALaDavid}, we first have to define the sets $H_i\subseteq E_i$ for $i\in [q]$. For this we use a variation of an idea by Adjiashvili~\cite{adjiashvili2017beating}, based on heavy edges.

We start by observing that for each tree $T_i=(V_i,E_i)$, where $i\in [q]$, there is a vertex $r_i\in V_i$ such that 
\begin{equation}\label{eq:rootInTi}
x\left(\cov(E_{T_i}(e,w))\setminus \cov(e)\right) \leq x\left(\cov(E_{T_i}(e,r_i))\setminus \cov(e)\right)
\quad \forall e=\{r_i,w\}\in E_i\enspace.
\end{equation}
To see why such a vertex $r_i$ exists, orient the edges $\{u,v\}\in E_i$ of $T_i$ as follows:
\begin{itemize}[nosep,topsep=0.2em]
\item if $x(\cov(E_{T_i}(e,u))\setminus \cov(e)) \leq x(\cov(E_{T_i}(e,v))\setminus \cov(e))$, orient $\{u,v\}$ from $u$ to $v$,
\item otherwise, orient $\{u,v\}$ from $v$ to $u$.
\end{itemize}
To obtain~\eqref{eq:rootInTi}, $r_i$ can be chosen to be any vertex with only outgoing edges, which always exists in a directed tree.
Notice that~\eqref{eq:rootInTi} implies the following
\begin{equation}\label{eq:lightAwayFromRoot}
\begin{aligned}
  x\left(\cov(E_{T_i}(e,v))\setminus \cov(e)\right) \leq 
  &x\left(\cov(E_{T_i}(e,u))\setminus \cov(e)\right)\\
&\quad\forall e=\{u,v\}\in E_i \text{ with $u$~}
\text{being on unique $s$-$v$ path in $T_i$.}
\end{aligned}
\end{equation}

The set $H_i$ will only contain so-called \emph{$\zeta$-heavy} edges, for $\zeta=\sfrac{\sqrt{k}}{4}$, which are edges $e\in E_i$ satisfying $x(\cov(e))\geq \zeta$. Let $U_i\subseteq E_i$ be all $\zeta$-heavy edges of $T_i$. Then $H_i\subseteq U_i$ is chosen to be all edges in the connected component of $(V_i,U_i)$ that contains $r_i$; Moreover, we denote by $W_i\subseteq V_i$ all vertices of the connected component of $(V_i,U_i)$ that contains $r_i$.

To see why point~\ref{item:thinSubInstances} of Lemma~\ref{lem:decompALaDavid} holds for the trees $T_i=(V_i,E_i)$ and sets $H_i\subseteq E_i$ for $i\in [q]$, consider the principal subtrees of $T_i/H_i$ obtained by using as root the vertex that corresponds to the contraction of $H_i$.
For each such principal subtree, there is an edge $e=\{u,v\}$ with $u\in W_i$ and $v\in V_i\setminus W_i$ such that the vertices of the principal subtree consist of $u$ and all vertices in the connected component of $(V_i,E_i\setminus \{e\})$ that contains $v$. We denote by $T_v$ this principal subtree.

Observe that $e=\{u,v\}$ is not $\gamma$-light (for otherwise, we would have split on that edge), and $e$ is by definition not $\zeta$-heavy. This implies
\begin{align*}
x\left(\cov(E_{T_i}(e,v))\setminus \cov(e)\right)
  \leq \frac{1}{\gamma} \cdot x(\cov(e))
  \leq \frac{\zeta}{\gamma}
  = \frac{k}{4}\enspace,
\end{align*}
where the first inequality follows from~\eqref{eq:lightAwayFromRoot} and $e$ not being $\gamma$-light, and the second one from the fact that $\{u,v\}$ is not $\zeta$-heavy.
This in turn implies 
\begin{equation}\label{eq:boundOnPrincipXVal}
x(\cov(E_{T_i}(e,v)))\leq x(\cov(E_{T_i}(e,v)\setminus \cov(e)) +x(\cov(e))
\leq \frac{k}{4} + \zeta \leq \frac{k}{2}\enspace.
\end{equation}
Because every leaf of $T_v$ must be (fractionally) covered by $x$, and any link can cover at most two leaves of $T_v$, we have that the number of leaves of $T_v$ is at most 
\begin{equation*}
2\cdot x\left(\cov(E_{T_i}(e,v))\right) \leq 2\cdot \frac{k}{2} = k\enspace,
\end{equation*}
where the inequality follows from~\eqref{eq:boundOnPrincipXVal}. This shows point~\ref{item:thinSubInstances} of Lemma~\ref{lem:decompALaDavid}.

Finally, observe that $\sfrac{1}{\zeta}\cdot x$ (fractionally) covers each $\zeta$-heavy edge at least once. Therefore, by applying a $2$-factor rounding algorithm to the vector induced by $\sfrac{1}{\zeta}\cdot x$, when we restrict ourselves to the TAP instance induced by the edges $\cup_{i=1}^q H_i$, we get a set of links $M\subseteq L$ covering $\cup_{i=1}^q H_i$ whose cost is at most $\sfrac{2}{\zeta}\cdot x(L) = O(\sfrac{1}{\sqrt{k}})\cdot x(L)$. This completes the proof by showing point~\ref{item:canCoverHi} of Lemma~\ref{lem:decompALaDavid}.

\end{proof}

\section{Derandomized Rewirings}\label{sec:derandomization}
In this section, we show that the randomized part of our approach can be completely derandomized. More precisely, we 
will sketch a proof of the following variant of Lemma \ref{lem:apx}:
\begin{lemma}\label{lem:wired-rounding-derandomized}
There exists a deterministic approximation algorithm for $O(1)$-wide $TAP$ instances with approximation guarantee 
$2(\sqrt{3}-1) < 1.465$.
\end{lemma} 

We stress that we choose to sketch the derandomization of the slightly weaker (in terms of approximation ratio) Lemma 
\ref{lem:apx} for the sake of simplicity. Nonetheless, the derandomization of Theorem \ref{thm:main} works in an almost 
identical way.

Let us review our rounding procedure. Clearly, the only randomized part is the algorithm in Lemma  
\ref{lem:wired-rounding}. In particular, given a TAP instance $(T=(V,E),L)$ where $T$ is a $k$-wide tree, and a 
 $k$-wide-LP solution $(x,\lambda)$, this algorithm samples a local solution $L_i^R$ for each principal 
 subtree $T_i$ and each $R\in\Lambda_i$ with probability $\lambda_i^R$. Let $L_i$ be the random set of links we sample 
for principal subtree  $T_i$, let $A_i$ be the random set of active nodes (i.e., a node in some $T_i$ that is an 
endpoint of a cross-link in $L_i$), and let $A$ be  the random variable $\cup_{i\in[q]} A_i$. Note that $A_i$ is 
completely determined by $L_i$. The key idea behind our  derandomization is the following: since $L_i$ has a domain 
whose size is polynomial in $|V|$, we can apply the method  of \emph{Conditional Expectations} in order to fix a local 
solution for each principal subtree deterministically.

More precisely, we begin by fixing a matching $M$ according to Lemma \ref{lem:matching}; remember that this matching 
is selected deterministically. Then, our rounding algorithm returns a tree augmentation whose size is at most
$$
\E\left[\sum_{i=1}^{q} |L_i|-\big|M\cap {
\begin{psmallmatrix}
A\\
2
\end{psmallmatrix}
}\big|\right].
$$
Clearly, we can always compute $\E[|L_i|]$.
Crucially, we observe that $L_i$ is independent of $L_j$ for any principal subtrees $T_i$ and $T_j$ (and hence the same 
is true for $A^i$ and $A^j$). This implies that we can compute $\E\left[\big|M\cap {
\begin{psmallmatrix}
A\\
2
\end{psmallmatrix}
}\big|\right]$, since we know 
$\Pr[u\in A_i, v\in A_j]$ for any principal subtrees $T_i$, $T_j$ and any vertices $u$, $v$ in them respectively.
Furthermore, since the domain size of any $L_i$ is polynomial for any principal subtree $T_i$, we can enumerate over 
its domain. In particular, we can condition on $L_1=L^{R_1}_1$, where 
$$
R_1=\argmin_{R\in \Lambda_1}\E\left[\sum\limits_{i\in[q]} |L_i|-\big|M\cap {
\begin{psmallmatrix}
A\\
2
\end{psmallmatrix}
}\big| \ \middle \vert L_1=L_1^{R_1}\right].
$$
By applying this reasoning inductively, we can condition on random variable $L_2$, $L_3$ etc., until we find 
some
$$
R_j=\argmin_{R\in \Lambda^j}\E\left[\sum\limits_{i\in[q]} |L_i|-\big|M\cap {
\begin{psmallmatrix}
A\\
2
\end{psmallmatrix}
}\big| \ \middle \vert L_1=L_1^{R_1},\dots, 
L_{j-1}=L^{R_{j-1}}_{j-1}\right],
$$
for all $j\in [q]$. Since the expected cost does not increase as we condition on a variable, we found a point in 
the sample space of variables $L_i$, for all $i\in[q]$, whose cost is at most the cost of 
$$
\E\left[\sum\limits_{i\in[q]} |L_i|-\big|M\cap {
\begin{psmallmatrix}
A\\
2
\end{psmallmatrix}
}\big|\right].
$$
By Lemmas~\ref{lem:wired-rounding} and~\ref{lem:matching} we thus obtain 
$$
\E\left[\sum\limits_{i\in[q]} |L_i|-\big|M\cap {
\begin{psmallmatrix}
A\\
2
\end{psmallmatrix}
}\big|\right]\leq x(\Lin)+2x(\Lcross)-\frac{x^2(\Lcr)}{|\Vcr|}\enspace,
$$
which is the same relation obtained for the randomized procedure, and thus the analysis in Lemma~\ref{lem:apx} applies, which immediately implies Lemma~\ref{lem:wired-rounding-derandomized}.

 \section{A Refined Approximation Factor}\label{sec:improved-apx}

In order to establish a stronger approximation guarantee than that of Lemma \ref{lem:apx}, we first provide a strengthening of Lemma~\ref{lem:matching}, described by the lemma below. Its proof is almost identical to the proof of Lemma~\ref{lem:matching}, except for a more careful application of the Cauchy-Schwarz inequality:
\begin{lemma}\label{lem:improved-matching}
	Let $G=(V,E)$ be an undirected graph, let $x\in\mathbb{R}_{\geq 0}^E$ such that for any $v\in V$ we have 
	$x(\delta(v)) \leq 1$, and let $A$ be a random subset of $V$ with a distribution that satisfies:
	\begin{enumerate}[nosep,label=(\roman*),itemsep=0.2em,topsep=0.2em]
		\item\label{sample1-improved}  $\Pr[v\in A]=x(\delta(v)) \quad \forall v\in V$, and 
		\item\label{sample2-improved}  $\Pr[v\in A\text{ and }u\in A]= \Pr[u\in A] \cdot \Pr[v\in A]\quad 
\forall \{u,v\}\in E$.
	\end{enumerate}
Moreover, let $z\in\mathbb{R}_{\geq 0}^E$ by a vertex of the polytope 
$Q = \left\{y\in \mathbb{R}^E_{\geq 0} \;\middle\vert\; y(\delta(v))= x(\delta(v)) \quad \forall 
v\in 
V\right\}$. 
Then, 
\begin{itemize}[nosep,topsep=0.2em]
	\item $|\supp(z)|\leq |V|$\enspace, and
	\item there exists a matching $M\subseteq E$ in $G$ such that
\begin{equation*}
	\E\left[ \big|
	M\cap {
\begin{psmallmatrix}
A\\
2
\end{psmallmatrix}
}
	\big|  \right] \geq \sum\limits_{e\in E} z_e^2\geq 
	\frac{(\sum_{e \in E_1} z_e)^2}{|E_1|}+\frac{(\sum_{e \in E_2} z_e)^2}{|E_2|}\enspace,
\end{equation*}
for any $E_1,E_2\subseteq 
	\supp(z)$ with $E_1\cap E_2=\emptyset.$
\end{itemize}
\end{lemma}

\begin{proof}
Clearly, $x\in Q$. Moreover, for any vertex $z$ of $Q$ we indeed have $|\supp(z)|\leq |V|$. This follows through the 
following standard sparsity argument: Any vertex $z$ of $Q$ is defined by $|E|$-many linearly independent and tight 
constraints; however, $Q$ has only $|V|$ constraints that are not nonnegativity constraints, and therefore, at least 
$|E|-|V|$ of the nonnegativity constraints must be tight at $z$, which implies $|\supp(z)| \leq |V|$.

We let the matching $M\subseteq E$ be a matching obtained by the greedy algorithm for maximum weight matchings with respect to 
the weights $z$. More formally, we order the edges $E=\{e_1,\ldots, e_m\}$ such that $z(e_1)\geq \ldots \geq z(e_m)$. 
We start with $M=\emptyset$ and go through the edges in the order $e_1,\ldots, e_m$. When considering $e_i$, we set 
$M=M\cup \{e_i\}$ if $M\cup \{e_i\}$ is a matching.

We have the following (which holds for any matching $M$):
\begin{equation}\label{eq:expMatchSize1-improved}
\E\left[\big|
	M\cap {
\begin{psmallmatrix}
A\\
2
\end{psmallmatrix}
}
	\big|  \right]
= \sum_{\{u,v\}\in M} x(\delta(u)) \cdot x(\delta(v))
= \sum_{\{u,v\}\in M} z(\delta(u)) \cdot z(\delta(v)) \enspace,
\end{equation}
which is an immediate consequence of property~\ref{sample1-improved} and~\ref{sample2-improved}, and of $z\in Q$. Next, we show that the following holds:
\begin{equation}\label{eq:expMatchSize2-improved}
\sum_{\{u,v\}\in M} z(\delta(u)) \cdot z(\delta(v)) \geq \sum_{e\in E} z(e)^2\enspace.
\end{equation}
Note that expanding the left-hand side of~\eqref{eq:expMatchSize2-improved} by using $z(\delta(u))=\sum_{e\in 
\delta(u)}z(e)$ and $z(\delta(v))=\sum_{e\in \delta(v)}z(e)$, leads to a sum of terms of the form $z(e)\cdot z(f)$ for 
some pairs $e,f\in E$. For any $e=\{u,v\}\in M$, the left-hand side of~\eqref{eq:expMatchSize2-improved} has a term of 
the form $z(e)^2$, due to the term $z(\delta(u))\cdot z(\delta(v))$. Now consider any non-matching edge $e=\{u,v\}\in 
E\setminus M$. Because $e$ is a non-matching edge, and we constructed $M$ by the greedy algorithm, there is at least one 
endpoint of $e$, say $u$, such that there is an edge $f\in M$ incident with $u$ that satisfies $z(f)\geq z(e)$. This 
implies that $z(\delta(u))\cdot z(\delta(v))$ contains a term $z(f)\cdot z(e) \geq z(e)^2$. Hence, for each $e\in E$, 
we identified a term on the left-hand side of~\eqref{eq:expMatchSize2-improved} of value at least $z(e)^2$; moreover, 
all such terms on the left-hand side are different, which implies~\eqref{eq:expMatchSize2-improved}.

The statement of the lemma now follows by combining~\eqref{eq:expMatchSize1-improved} 
and~\eqref{eq:expMatchSize2-improved} with the 
following inequality, which is valid for any $E_1,E_2\subseteq \supp(z)$ such that $E_1\cap 
E_2=\emptyset$, and 
follows by the Cauchy-Schwarz inequality:
\begin{equation*}
\sum_{e\in E} z(e)^2\geq \sum_{e\in E_1} z(e)^2+\sum_{e\in E_2} z(e)^2\geq 
\frac{z(E_1)^2}{|E_1|}+\frac{z(E_2)^2}{|E_2|}\enspace.
\end{equation*}
\end{proof}

Given the above lemma, we are now ready to prove the randomized version of Theorem \ref{thm:main}:
\begin{proof}[Proof of Theorem \ref{thm:main}]
Given a $k$-wide TAP instance $(T,L)$ with $k=O(1)$, and a solution $(x,\lambda)$ to the $k$-wide-LP, by Lemma~\ref{lem:wired-rounding} we can randomly round $(x,\lambda)$ to a solution $\mathcal{B}$ such that 
\begin{equation}\label{eq:solBBound-improved}
\E[|\mathcal{B}|] \leq x(\Lin)+2x(\Lcross)-\E_{A}\left[\eta \left(A,\Lcr\cap
\begin{psmallmatrix}
A\\
2
\end{psmallmatrix}
\right)\right]\enspace,
\end{equation}
where $A$ fulfills the conditions of Lemma~\ref{lem:wired-rounding}.

We apply Lemma~\ref{lem:improved-matching} to the graph $H=(\Vcr, \Lcr)$ and the restriction of $x$ to $\Lcr$, where, as usual, $A\subseteq \Vcr$ will be the random set of active nodes. To apply Lemma~\ref{lem:improved-matching}, we define $z\in \mathbb{R}_{\geq 0}^{\Lcr}$ to be any vertex of the polytope
\begin{equation*}
Q\coloneqq\left\{y\in \mathbb{R}_{\geq 0}^{\Lcr} \;\middle\vert\; y(\delta_H(v)) = x(\delta_H(v)) \;\;\forall v\in \Vcr\right\}\enspace.
\end{equation*}
Lemma~\ref{lem:improved-matching} then implies that
\begin{equation}\label{eq:boundForSuppZ}
|\supp(z)| \leq |\Vcr|\enspace,
\end{equation}
and the expected size of a maximum cardinality matching in $(A,\Lcr\cap
\begin{psmallmatrix}
A\\
2
\end{psmallmatrix}
)$ is lower bounded by
\begin{equation*}
\E_A\left[\eta(A, \Lcr\cap
\begin{psmallmatrix}
A\\
2
\end{psmallmatrix}
)
\right] \geq \frac{z(E_1)^2}{|E_1|} + \frac{z(E_2)^2}{|E_2|}\enspace,
\end{equation*}
for any $E_1, E_2\subseteq \supp(z)$ with $E_1\cap E_2=\emptyset$.

Let $\Ltwol$ be the links in $\Lcr \cap \supp (z)$ whose endpoints are both leaves, and $\Lonel$ be the remaining links in $\Lcr\cap \supp(z)$. In what follows, we set $E_1 = \Lonel$ and $E_2=\Ltwol$, and will use inequality~\eqref{eq:strongCSforZ} for these two sets, thus leading to
\begin{equation}\label{eq:strongCSforZ}
\E_A\left[\eta(A, \Lcr\cap
\begin{psmallmatrix}
A\\
2
\end{psmallmatrix}
)
\right] \geq \frac{z(\Ltwol)^2}{|\Ltwol|} + \frac{z(\Lonel)^2}{|\Lonel|}\enspace.
\end{equation}
Moreover, because $z$ is a vertex of $Q$, we have
\begin{equation}\label{eq:LtwolBound}
|\Ltwol| \leq K\enspace,
\end{equation}
where $K$ is the number of leaves of $T$. This follows by the fact that, except for nonnegativity constraints, there are only $K$ constraints in $Q$ that involve links with both endpoints being leaves.

\medskip

Consider $z'\in\mathbb{R}^L$ defined by
$$z'_\ell=
\begin{cases}
z_\ell & \text{for } \ell\in\Lcr\enspace,\\
x_\ell & \text{otherwise}\enspace.
\end{cases}
$$
Observe that $z'\in \Pi_T$.
We define $\phi\in [0,1]$ so that $z'(\Lonel)=\phi \cdot z'(\Lcr)=\phi\cdot x(\Lcr)$. (Hence $z'(\Ltwol)=(1-\phi) \cdot 
z'(\Lcr)=(1-\phi)\cdot x(\Lcr)$.)

We proceed by deriving an upper bound for $\E[|\mathcal{B}|]/\OPT^*$, where $\OPT^* = x(L)=z'(L)$ is the optimal value of the $k$-wide-LP, as usual. Notice that $|\Vcr|< 2K$, which readily follows from the fact that the average degree of $T$ must be strictly below $2$, because $T$ is a tree.
Moreover, we must have
\begin{equation*}
\sum_{\ell\in L:v\in \ell}z'_{\ell}\geq 1 \quad \text{for each leaf node $v\in V$}
\end{equation*}
for the edge incident with $v$ to be covered by the LP solution $z'\in \Pi_T$. Each link can contribute to at most two of the sums above, for two different leaf nodes, and the links in $\Lup\cup \Lncr\cup \Lonel$ to at most one of these sums. Thus we can conclude that
\begin{equation}\label{eq:leavesBound}
2z'(\Lcross)+2z'(\Lin)-z'(\Lup)-z'(\Lncr)-z'(\Lonel)\geq K\geq |\Vcr|/2\enspace.
\end{equation}
Let $\ain = \sfrac{x(\Lin)}{\OPT^*} = \sfrac{z'(\Lin)}{\OPT^*}$ be the fraction of $\OPT{}^*$ that corresponds to links in $\Lin$. We define analogously $\aup$, $\across$, $\acr$, $\ancr$, $\atwol$, and $\aonel$. By Lemma \ref{lem:matching}, the expected size of the maximum cardinality matching in $(A, \Lcr\cap
\begin{psmallmatrix}
A\\
2
\end{psmallmatrix}
)$ satisfies
\begin{equation}\label{eq:expMatchSizeApx-improved}
\begin{aligned}
\E_{A}\left[\eta(A, \Lcr\cap
\begin{psmallmatrix}
A\\
2
\end{psmallmatrix}
)
\right] \geq \frac{x(\Lcr)^2}{|\Vcr|}=\frac{z'(\Lcr)^2}{|\Vcr|}
 \geq \frac{
(\acr)^2}{4-2\aup-2\ancr-2\aonel}\cdot \OPT^*\enspace,
\end{aligned}
\end{equation}
where the last inequality above follows from~\eqref{eq:leavesBound}.
Combining~\eqref{eq:solBBound-improved} and~\eqref{eq:expMatchSizeApx-improved}, we conclude that
\begin{equation}\label{eq:guaranteeOfB-improved1}
\begin{aligned}
\frac{\E[|\mathcal{B}|]}{\OPT{}^*} &\leq  
\ain+2\across-\frac{(\acr)^2}{4-2\aup-2\ancr-2\aonel}\\
  &= 1+\across
-\frac{(\acr)^2}{4-2\aup-2\ancr-2\phi\acr}\enspace.
\end{aligned}
\end{equation}

Observe that $|\Ltwol|+|\Lonel| = |\supp(z)| \leq |\Vcr| \leq 2K$, where the first inequality follows from~\ref{eq:boundForSuppZ}. For simplicity, we define $q \coloneqq K - |\Ltwol|$, and notice that $q\geq 0$ due to~\eqref{eq:LtwolBound}. Hence, $|\Ltwol| = K + q$ and $|\Lonel| \leq 2K - |\Ltwol| = K+q$.
We thus obtain by~\eqref{eq:strongCSforZ}
\begin{equation}\label{eqn:refinedMatch}
\begin{aligned}
\E_{A}\left[\eta(A, \Lcr\cap
\begin{psmallmatrix}
A\\
2
\end{psmallmatrix}
)
\right] \geq \frac{z'(\Ltwol)^2}{K-q}+\frac{z'(\Lonel)^2}{K+q}
 = 
 \frac{((1-\phi)z'(\Lcr))^2}{K-q}+\frac{(\phi z'(\Lcr))^2}{K+q}\enspace.
\end{aligned}
\end{equation}

For $\phi\leq 0.5$, the right-hand side of \eqref{eqn:refinedMatch} is growing in $q$, hence the worst matching size 
cost is achieved for $q=0$ in this case, and is at least
$$
\frac{1}{K}\cdot \left(((1-\phi)z'(\Lcr))^2+(\phi z'(\Lcr))^2\right)\enspace,
$$
which implies together with~\eqref{eq:solBBound-improved} the following inequality for any $\phi\leq 0.5$:
\begin{equation}\label{eq:guaranteeOfB-improved2}
\frac{\E[|\mathcal{B}|]}{\OPT{}^*}\leq 
1+\across-\frac{((1-\phi)\acr)^2+(\phi \acr)^2}{2-\aup-\ancr-\phi\acr}\enspace.
\end{equation}

Finally, from Corollary~\ref{cor:chvatal-rounding} we know that we can round $x$ to a solution $\mathcal{A}$ such that
\begin{equation}\label{eq:guaranteeOfA-improved}
\frac{|\mathcal{A}|}{\OPT{}^*}\leq
2\ain-\aup+\across = 2-\across-\aup\enspace,
\end{equation}
 where the equality follows from $\across+\ain =1$.
 
 Overall, from~\eqref{eq:guaranteeOfA-improved},~\eqref{eq:guaranteeOfB-improved1}, and~\eqref{eq:guaranteeOfB-improved2}, we get that if $\phi\leq 0.5$, the approximation ratio is at most
 \begin{align}
 1+\min \Bigg\{ & 1-\across-\aup,\nonumber\\
 & \across-\frac{(\acr)^2}{4-2\aup-2\ancr-2\phi \acr},\nonumber\\
 & \across-\frac{((1-\phi)\acr)^2+(\phi\acr)^2}{2-\aup-\ancr-\phi \acr}\Bigg\}\enspace,\label{eq:min1}
 \end{align}
 and if $\phi\geq 0.5$ it is at most
  \begin{align}
  1+\min\Bigg\{ & 1-\across-\aup,\nonumber\\
  & \across-\frac{(\acr)^2}{4-2\aup-2\ancr-2\phi \acr}\Bigg\}\enspace.\label{eq:min2}
  \end{align}
  First, observe that all terms inside the minima in (\ref{eq:min1}) and (\ref{eq:min2}) are non-increasing in $\aup$ 
and $\ancr$. Therefore, the maxima of both (\ref{eq:min1}) and (\ref{eq:min2}) over all values 
$\across,\aup,\ancr,\acr\in[0,1]$ have $\aup=\ancr=0$. Hence, we focus on finding the maxima of~\eqref{eq:min1} and~\eqref{eq:min2} for $\aup=\ancr=0$. Furthermore, if $\phi\geq 0.5$, then~\eqref{eq:min2} is non-increasing in $\phi$, which means we can also assume $\phi\leq 0.5$.
  
  Therefore, we conclude the ratio is at most 
   \begin{align*}
   1+\min_{\across\in[0,1],\phi\in[0,1/2]}\Bigg\{ & 1-\across,\\
   & \across-\frac{(\across)^2}{4-2\phi \across},\\
   & \across-\frac{((1-\phi)\across)^2+(\phi\across)^2}{2-\phi \across}\Bigg\}\enspace.
   \end{align*}
   
   Finally, observe that in this range, the third term is always dominated by the second one because $1-2\phi+2\phi^2 \geq 0.5$. Therefore, the approximation ratio is at most
      \begin{align*}
      1+\min_{\across\in[0,1],\phi\in[0,1/2]}\Bigg\{ & 1-\across,\\
      & \across-\frac{((1-\phi)\across)^2+(\phi\across)^2}{2-\phi \across}\Bigg\}\enspace.
      \end{align*}
We finish the proof by showing that the maximum ratio attainable under these constraints is $\sqrt{34}/4 < 1.458$.

\medskip

	Let $f_1(x,y)\coloneqq 1-x$, $f_2(x,y)\coloneqq x-x^2\cdot \frac{1-2y+2y^2}{2-xy}$, and $f(x,y)\coloneqq \min\{f_1(x,y),f_2(x,y)\}$. In the sequel, 
	we will only consider $(x,y)\in[0,1]\times[0,\frac{1}{2}]$, since our goal is to find 
	$$
	\max_{(x,y)\in[0,1]\times[0,\frac{1}{2}]} f(x,y)\enspace.
	$$
	
	Let $D=[0,1]\times[0,\frac{1}{2}]$. Consider any maximizer $(x^*,y^*)$ of $f$ over $D$. We will show that $(x^*,y^*)$ satisfies $f_1(x^*,y^*) = f_2(x^*,y^*)$.

To this end, observe that if $f_1(x^*,y^*)<f_2(x^*,y^*)$, then $(x^*,y^*)$ is a local maximizer of $f_1$. This directly implies that $x^*=0$, which holds for any local maximizer of $f_1$ in $D$. However, $f_1(0,y)=1>f_2(0,y)$ for any $y\in [0,\frac{1}{2}]$, which is a contradiction.

Conversely, consider the case $f_1(x^*,y^*) > f_2(x^*,y^*)$. First of all, observe that 
	$$
	\frac{\partial}{\partial x}f_2(x,y)=1-(1-2y+2y^2)\cdot \frac{2x(2-xy)+yx^2}{(2-xy)^2}\enspace.
	$$
	Now, observe that $f(x^*,y^*) \geq f_2(\frac{1}{2},\frac{1}{2})\geq 0.4$. Since $f_1(x^*,y^*)> f_2(x^*,y^*)\geq 0.4$, it 
	follows that $x^*\leq 0.6$. Furthermore, for $x\leq 0.6$, we have
\begin{equation*}
	\frac{\partial}{\partial x}f_2(x,y)=1-(1-2y+2y^2)\cdot \frac{4x-x^2y}{(2-xy)^2}\geq 
	1-\frac{\frac{12}{5}-\frac{9}{25}y}{1.7^2}\geq 	1-\frac{\frac{12}{5}}{1.7^2}>0\enspace.
\end{equation*}
	Hence, $f_2(x,y)$ is strictly increasing in $x$ for $x\leq 0.6$, which implies $x^*\geq 0.6$. Moreover, because $\frac{\partial}{\partial x} f_2(x,y)$ and $f_1(x,y)$ are continuous, we even get $x^*>0.6$. However, this contradicts $x^*\leq 0.6$.
Therefore, we conclude that for any maximizer $(x^*,y^*)$ of $f$, $f_1(x^*,y^*)=f_2(x^*,y^*)$.

	Now, we wish to solve $f_1(x,y)=f_2(x,y)$ for $x$. We have:
\begin{align*}
1-x&=x-x^2\cdot \frac{1-2y+2y^2}{2-xy}\enspace,\\
\intertext{which implies}
0&=x^2(1+2y^2)-x(4+y)+2\enspace,\\
\intertext{and thus leads to}
x=&\frac{4+y\pm\sqrt{(4+y)^2-8(1+2y^2)}}{2(1+2y^2)}\enspace.
\end{align*}
Since $x\in[0,1]$, we need to pick 
\begin{equation}\label{eq:xAsFuncOfY}
x=\frac{4+y-\sqrt{(4+y)^2-8(1+2y^2)}}{2(1+2y^2)}\enspace.
\end{equation}
Finally, we want to maximize $f_1(x,y)=f_2(x,y) = 1-x$, which is equal to the following:
\begin{align*}
&1-\frac{4+y-\sqrt{(4+y)^2-8(1+2y^2)}}{2(1+2y^2)}\\
=&1-\frac{1}{2(1+2y^2)}\frac{(4+y)^2-\left(\sqrt{(4+y)^2-8(1+2y^2)}\right)^2}{4+y+\sqrt{(4+y)^2-8(1+2y^2)}}\\
=&1-\frac{8(1+2y^2)}{2(1+2y^2)(4+y+\sqrt{(4+y)^2-8(1+2y^2)})}\\
=&1-\frac{4}{4+y+\sqrt{(4+y)^2-8(1+2y^2)}}\enspace.&
\end{align*}
Equivalently, we want to maximize $h(y)\coloneqq 4+y+\sqrt{(4+y)^2-8(1+2y^2)}$. We have
	$$
	\frac{\partial}{\partial y}h(y)=1+\frac{1}{2}\frac{8-30y}{\sqrt{(4+y)^2-8(1+2y^2)}}\enspace.
	$$
	Setting $\frac{\partial}{\partial y}h(y)$ to $0$, we finally obtain
	\begin{equation*}
	30y^2-16y+1=0\enspace,
  \end{equation*}
and thus
\begin{equation*}
	y=\frac{1}{30}\cdot \left({8 \pm \sqrt{34}}\right)\enspace.
\end{equation*}
One can conclude that the value of $y$ must be set to $y=\sfrac{(8 + \sqrt{34})}{30}$, because by setting $\frac{\partial}{\partial y}h(y)$ to $0$, 
we equated $8-30y$ (which is negative for $y=\sfrac{(8-\sqrt{34})}{30})$ to a square root of a positive real number. By 
checking that $h(0)$ and $h(1)$ are smaller than $h(\sfrac{(8+\sqrt{34})}{30})$, we determine that $h(y)$ is indeed maximized for $y^*=\sfrac{(8 + \sqrt{34})}{30}$. By determining $x^*$ from $y^*$ through~\eqref{eq:xAsFuncOfY}, we obtain $x^*= 2 - \sfrac{\sqrt{34}}{4}$, and the resulting approximation guarantee is therefore 
\begin{equation*}
1 + f(x^*,y^*) = 1 + f_1(x^*,y^*) = 2 - x^* = \frac{\sqrt{34}}{4} < 1.458\enspace,
\end{equation*}
as desired.
\end{proof}

\section{Extension to Weighted TAP}\label{sec:WTAP}

In this section, we show how our techniques can be extended to weighted TAP instances $(T,L,c)$, where $c\in \mathbb{R}^L_{>0}$ are the
link costs, if the ratio $c_{\max}/c_{\min}$ between the maximum cost and minimum cost is bounded by some constant $\Delta \geq 1$.
More precisely, we show that one can beat the factor $1.5$ in this setting, by a constant that depends on $\Delta$.
To do so, we establish the following analogues of Theorems~\ref{thm:decomp} and~\ref{thm:mainKWide} for this setting:
\begin{theorem}\label{thm:decomp-weights}
Let $k,\Delta\in \mathbb{Z}_{\geq 1}$, $\alpha \geq 1$, and $\mathcal{A}$ be an algorithm that is an 
$\alpha$-approximation for weighted TAP on $k\Delta$-wide instances with $c_{\max}/c_{\min}\leq \Delta$. Then 
there is an $(\alpha+O(\sfrac{1}{\sqrt{k}}))$-approximation algorithm for weighted TAP with $c_{\max}/c_{\min}\leq 
\Delta$, 
making a polynomial number of calls to $\mathcal{A}$ and performing further operations that take time polynomially 
bounded in the input size and $\Delta$.
This reduction also works for randomized algorithms, where the approximation guarantees hold in expectation.
\end{theorem}

\begin{theorem}\label{thm:mainKWide-weights}
Let $k,\Delta\in \mathbb{Z}_{\geq 1}$. There exists an approximation algorithm for $k\Delta$-wide 
weighted TAP instances $(T=(V,E),L,c)$ with $c_{\max}/c_{\min}\leq \Delta$, whose running time is polynomial in 
$|V|^{k\Delta}$, and whose approximation guarantee is $\frac{3}{2}-g(\Delta)$, for 
some positive function $g:[1,\infty)\rightarrow \mathbb{R}_{>0}$ that tends to $0$ as $\Delta$ tends to infinity.
\end{theorem}

For the sake of simplicity, and since the proofs of the above results bear strong similarities with their respective 
proofs in the unweighted setting, we will try to sketch their proofs by stressing the points where our approach for the 
weighted case diverges from our approach in the unweighted one. A key point we want to highlight is that the 
decomposition behind Theorem~\ref{thm:decomp-weights} produces $k\Delta$-wide trees, but still both the decomposition 
and the rounding algorithm run in polynomial time, since $\Delta$ is a constant. Concerning the approximation 
guarantee we achieve, we remark that, in the interest of clarity, we make no effort to obtain the best possible function $g(\Delta)$, and instead focus on showing that some strictly positive function $g(\Delta)$ exists as described in Theorem~\ref{thm:decomp-weights}, which is all we need to obtain that, if $c_{\max}/c_{\min}$ is a constant, then we can improve upon the $\sfrac{3}{2}$-approximation factor.

Let us begin with establishing Theorem~\ref{thm:decomp-weights}. From the discussion in 
Section~\ref{sec:reductionToWideTrees}, and from inspecting the proofs of Lemma~\ref{lem:boostedDecomp} and Theorem 
\ref{thm:decomp}, it follows that we only need to extend Lemma~\ref{lem:decompALaDavid} to the weighted setting with bound $\Delta$ on the
ratio between highest link cost and lowest link cost. Without loss of generality, we can assume that
$c_{\min}=1$, by scaling the link costs. This implies $c_{\max}\leq \Delta$.
We sketch the proof of a variant of Lemma~\ref{lem:decompALaDavid}, tailored for the bounded-weights case:

\begin{lemma}\label{lem:decompALaDavid-weights}
Let $k\in \mathbb{Z}_{\geq 1}$, let $(T,L,c)$ be a weighted TAP instance on a tree $T=(V,E)$ with $c_{\ell}\in[1,\Delta]$ for all $\ell\in
L$, and let $x\in P_G$.
Then one can efficiently partition $G$ into subtrees $G_i=(V_i,E_i)$ for $i\in [q]$ and find $H_i\subseteq E_i$ for 
$i\in [q]$ such that:
\begin{enumerate}[label=(\roman*),itemsep=-0.2em,topsep=0.4em]
\item\label{item:thinSubInstances-weight} For $i\in [q]$, $G_i/H_i$ is $k\Delta$-wide.
\item\label{item:canCoverHi-weight} One can efficiently obtain a link set $M\subseteq L$ with $c(M) = O(\sfrac{1}{\sqrt{k}}) \sum_{\ell\in L} c_\ell x_\ell$ covering $\cup_{i=1}^q H_i$.
\item\label{item:goodLpSolForSubProbs-weight} $\sum_{i=1}^q \sum_{\ell \in \cov(E_i)}c_\ell x_\ell =  
(1+O(\sfrac{1}{\sqrt{k}}))\cdot \sum_{\ell\in L} c_\ell x_\ell$.
\end{enumerate}
\end{lemma}
\begin{proof}[Proof-sketch of Lemma \ref{lem:decompALaDavid-weights}]
 We observe that, by modifying the definition of $\gamma$-light edges (Definition \ref{def:light}) to be any edge $e=\{u,v\}$ such that
$$
\sum_{l\in \cov(e)}c_{\ell} x_{\ell} \leq \gamma \cdot \min \left\{ \sum_{\ell\in \cov(E(e,u))\setminus\cov(e)}c_{\ell} x_{\ell},\sum_{\ell\in 
\cov(E(e,v))\setminus\cov(e)}c_{\ell} x_{\ell} \right\},
$$
then, by applying the decomposition process of Section~\ref{sec:decomp} and following the proof of Lemma~\ref{lem:decompALaDavid}, point~\ref{item:goodLpSolForSubProbs-weight} follows immediately. 
Point~\ref{item:canCoverHi-weight} will follow simply from our choice of the constant $\zeta$ with respect to the 
definition of $\zeta$-heavy edges.
Finally, for point~\ref{item:thinSubInstances-weight}, it suffices to observe that
\begin{equation*}
\sum_{\ell\in\cov(e)} c_\ell x_\ell \leq \Delta x(\cov(e))\qquad \forall e=\{u,v\}\in E \enspace,
\end{equation*}
which implies in particular that for any edge $e$ that is not 
$\zeta$-heavy,
\begin{equation*}
\sum_{\ell\in\cov(e)} c_\ell x_\ell \leq \Delta\zeta\enspace.
\end{equation*}
Finally, observe that for any edge $e$ that is adjacent to a leaf we have
\begin{equation*}
\sum_{\ell\in\cov(e)} c_\ell x_\ell \geq 1\enspace.
\end{equation*}
Following the proof of Lemma \ref{lem:decompALaDavid} using these updated bounds then establishes points 
\ref{item:canCoverHi-weight} and 
\ref{item:thinSubInstances-weight}.
\end{proof}

With the previous result at hand, our approach to approximating weighted TAP with bounded cost ratio on $k\Delta$-wide trees is very
similar to our approach for the unweighted case. Ideally, we would like to use the weighted version of the $k\Delta$-wide-LP, and apply
exactly the same rounding algorithm as the one we presented in Section \ref{sec:algo}: either round
using the  $\{0,1/2\}$-Chv\'atal-Gomory cuts, or sample a solution for each principal subtree and apply rewiring. The rounding based on Chv\'atal-Gomory cuts does not require unweighted instances. The only additional technical problem appears in the rewiring step for the latter rounding procedure. More precisely, rewiring as introduced previously may actually increase costs, because it is oblivious to link costs. In particular, the rewiring process may rewire two links into one link of larger total cost; notice that this will never happen in the unweighted setting, and this fact is crucial in proving that rewiring works for that case.

In order to remedy this situation, we conduct the rewiring process in a way that guarantees that every time we rewire a pair of
links into a single link, we decrease the cost of our rounded solution by some significant factor. Towards this goal, we will split the
cross-links into weight groups, where the costs of links in the same group differ only by a small constant factor. Then, instead of applying the rewiring process to
all the cross-links, we will apply it independently to each weight group. While this modification will decrease the number of rewirings we
will be able to perform, it will ensure that replacing two links by one decreases the cost of our solution. For our argument to be
complete, we will only need to ensure that there exists at least one group such that the links it contains represent a significant
fraction of the LP cost. We will show that this is true as long as the max-to-min cost ratio is a constant, which will imply that it suffices to create only a constant number of groups, and therefore at least one of those groups will contain a constant fraction of the cross-link weights in our LP solution.

Let us now establish Theorem \ref{thm:mainKWide-weights}:
\begin{proof}[Proof-sketch of Theorem \ref{thm:mainKWide-weights}]
Let $(T,L,c)$ be a weighted TAP instance, where $T$ is a $k\Delta$-wide tree and $c_{\ell}\in [1,\Delta]$ for any $\ell\in L$. Moreover, let $(x,\lambda)$ by an optimal solution to the weighted version of the corresponding $k\Delta$-wide-LP, whose cost will be denoted by $\OPT^*\coloneqq \sum_{\ell\in L}c_\ell x_\ell$. 

We start by discussing the rounding procedure based on Chv\'atal-Gomory cuts, as introduced by Fiorini et 
al.~\cite{fiorini_2018_approximating}. We first remark that Corollary~\ref{cor:chvatal-rounding} extends 
to the weighted setting---and was indeed used in this setting in~\cite{fiorini_2018_approximating}---which readily 
follows from Theorem~\ref{thm:fioriniGCexact}. Hence, we know that we can round $x$ to a solution $\mathcal{A}$ with 
approximation guarantee 
\begin{equation}\label{eq:guarOfADelta}\begin{aligned}
\frac{1}{\OPT^*} \cdot\hspace{-0.25em}\sum_{\ell\in\mathcal{A}} c_\ell x_\ell &\leq
\frac{1}{\OPT^*}\cdot\left(2\cdot \sum_{\ell\in\Lin} c_\ell x_\ell -\hspace{-0.25em}\sum_{\ell\in\Lup} c_\ell x_\ell +\hspace{-0.25em}\sum_{\ell\in\Lcross} c_\ell x_\ell \right)\\
&= 2\ain - \aup  + \across\\
&=2-\across - \aup\enspace,
\end{aligned}\end{equation}
where the second equality follows from $\across+\ain =1$.

\medskip

We now discuss our (slightly adapted) rounding procedure based on rewirings.
To this end, we partition $\Lcr$ into $p=\lceil\log_{1.5}\Delta\rceil$ groups $\{G_1,\ldots,G_p\}$, such that for any 
$h\in[p]$ and any $\ell,\ell'\in G_h$, we have $c_\ell\leq \frac{3}{2}c_{\ell'}$. Let 
\begin{equation*}
j=\argmax_{h\in[p]} x(G_h) \enspace,
\end{equation*}
which implies 
\begin{equation}\label{eq:G_j}
x(G_j) \geq \frac{1}{p} \sum_{h=1}^p x(G_h) = \frac{1}{p} x(\Lcr)\enspace.
\end{equation}
We perform our rounding step as in the unweighted case, with the only difference being that we apply the rewiring 
procedure separately to each link set $G_h$ for $h\in [p]$. First, observe that rewiring will always decrease the 
cost of the rounded solution, because for any two links $\ell,\ell'$ in the same group $G_h$, if we replace $\ell$ and 
$\ell'$ with a new link $\ell^*\in G_h$, the change in cost is
\begin{equation}\label{eq:costGainWRew}
c(\ell^*) - c(\ell) - c(\ell') \leq c(\ell^*) - \frac{4}{3}c(\ell^*) = -\frac{1}{3}c(\ell^*) \leq -\frac{1}{3}\enspace,
\end{equation}
where we used that both links $\ell$ and $\ell'$ have cost at least $\frac{2}{3}c(\ell^*)$, because they are in the same 
group $G_h$ as $\ell^*$; moreover, the last inequality follows from the fact that all link costs are within 
$[1,\Delta]$. Hence, for each pair of links that gets rewired, the cost decreases by at least $\sfrac{1}{3}$.

For simplicity of exposition, we analyze in the sequel only the cost gain we get through rewiring in the group $G_j$. 
This is enough to prove Theorem~\ref{thm:mainKWide-weights}.
As before, we sample a local solution $L_i$ for each principal subtree $T_i$ and compute a maximum matching $M$ among active vertices, which is then used for rewiring in $G_j$. We highlight that, because we only consider links in $G_j$, active vertices are critical vertices that are adjacent to a link in $G_j$, and not just any link in $\Lcr$.
By the discussion in Section~\ref{sec:algo}, it follows that we get a solution $\mathcal{B}$ with the following guarantee---analogous to~\eqref{eq:boundCardCalL}---when applying our rounding approach based on rewiring.
\begin{equation*}
\sum_{\ell \in \mathcal{B}} c_\ell \leq \sum_{i=1}^q c(L_i) - \frac{1}{3}|M|\enspace,
\end{equation*}
where we used the fact that any rewiring step gains at least $\sfrac{1}{3}$ in terms of costs due to~\eqref{eq:costGainWRew}, and there is one rewiring step for each edge in the matching $M$.
We thus get
\begin{equation}\label{eq:guarBRew}
\E\left[\sum_{\ell\in \mathcal{B}} c_\ell \right]
  \leq \sum_{\ell\in \Lin} c_\ell x_\ell + 2 \sum_{\ell\in \Lcross} c_\ell x_\ell - \frac{1}{3}
\E_A\left[\eta\left(A, G_j \cap
\begin{psmallmatrix}
A\\
2
\end{psmallmatrix}
\right) \right]\enspace,
\end{equation}
where the set of active vertices $A$ fulfills the distributional properties described by~\eqref{eq:propsOfA}. Hence, we can apply Lemma~\ref{lem:matching} to obtain
\begin{equation*}
\E_A\left[\eta\left(A, G_j \cap
\begin{psmallmatrix}
A\\
2
\end{psmallmatrix}
\right) \right] \geq \frac{x(G_j)^2}{|\Vcr|}\enspace,
\end{equation*}
which, combined with~\eqref{eq:guarBRew}, leads to the following expected
approximation guarantee for $\mathcal{B}$:
\begin{equation}\label{eq:guarBRewEx}
\begin{aligned}
\frac{1}{\OPT^*} \cdot \E\left[\sum_{\ell\in \mathcal{B}} c_\ell \right]
  &\leq \ain + 2 \across - \frac{x(G_j)^2}{3 |\Vcr| \OPT^*}\\
&= 1+ \across - \frac{x(G_j)^2}{3 |\Vcr| \OPT^*} \\
&\leq 1 + \across - \frac{x(\Lcr)^2}{3 p^2 |\Vcr| \OPT^*} \enspace,
\end{aligned}
\end{equation}
where the last inequality follows from~\eqref{eq:G_j}. As in the unweighted case, we can bound the number of critical vertices as follows (see~\eqref{eq:leaves-vs-x}).
\begin{equation}
2x(\Lcross)+2x(\Lin)-x(\Lup)-x(\Lncr) = 2x(L) - x(\Lup) - x(\Lncr) \geq K > \frac{|\Vcr|}{2}\enspace.
\end{equation}
Together with~\eqref{eq:guarBRewEx}, we thus obtain
\begin{equation}\label{eq:guarBRewEx2}
\frac{1}{\OPT^*}\cdot \E\left[\sum_{\ell\in \mathcal{B}} c_\ell\right]
\leq 1 + \across - \frac{1}{6 p^2}\cdot \frac{
\left(\frac{x(\Lcr)}{\OPT^*}\right)^2
}{
\frac{2x(L) - x(\Lup) - x(\Lncr)}{\OPT^*}
}\enspace.
\end{equation}
To further expand~\eqref{eq:guarBRewEx2}, we first observe that
\begin{equation}\label{eq:upCrossRel}
x(\Lup) \geq x(\Lncr)\enspace,
\end{equation}
which follows by Lemma~\ref{lem:up}, whose proof does not rely on links having unit costs. (Only the second implication of the lemma, i.e., $\aup\geq \ancr$, relies on links being unweighted.)
Moreover, we have the following relations stemming from the fact that all link costs are within $[1,\Delta]$.
\begin{align}
\frac{x(L)}{\OPT^*} &\leq 1\enspace,\\
\intertext{because every link has cost at least $1$, and}
\frac{x(\Lcr)}{\OPT^*} &\geq \frac{1}{\Delta} \acr\enspace,\label{eq:ccrToA}\\
\frac{x(\Lncr)}{\OPT^*} &\geq \frac{1}{\Delta} \ancr\enspace,\label{eq:nccrToA}
\end{align}
because every link has cost at most $\Delta$. Using~\eqref{eq:upCrossRel}--\eqref{eq:nccrToA}, we can further expand the (expected) approximation guarantee of $\mathcal{B}$ given by~\eqref{eq:guarBRewEx2} as follows.
\begin{align*}
\frac{1}{\OPT^*} \cdot \E\left[ \sum_{\ell\in \mathcal{B}} c_\ell \right]
  &\leq 1 + \across - \frac{1}{12 p^2 \Delta^2}\cdot
    \frac{(\acr)^2}{1-\frac{1}{\Delta}\ancr}\\
  &\leq 1 + \across - \frac{1}{12 p^2 \Delta^2} \cdot (\acr)^2\\
  &= 1 + \across - \frac{1}{12 p^2 \Delta^2} \left(\across - \ancr\right)^2\enspace,
\end{align*}
where the last equality follows from $\across = \acr + \ancr$. Hence, the (expected) approximation factor obtained by returning the better solution of $\mathcal{A}$ and $\mathcal{B}$ is, by the above inequality and~\eqref{eq:guarOfADelta}, upper bounded by 
\begin{equation}\label{eq:apxWTAPcomb}
1 + \min\left\{1-\across - \frac{1}{\Delta}\ancr, \ \across - \frac{1}{12p^2\Delta^2} \left(\across - \ancr \right)^2 \right\}\enspace,
\end{equation}
where we used $\aup \geq \sfrac{\ancr}{\Delta}$ to further expand the guarantee given by~\eqref{eq:guarOfADelta}, which 
follows from~\eqref{eq:upCrossRel} and the fact that link costs are within $[1,\Delta]$.

One way to provide an upper bound for the approximation guarantee described by~\eqref{eq:apxWTAPcomb}, is to analyze the two cases $\ancr\leq \frac{1}{2}\across$ and $\ancr > \frac{1}{2}\across$ separately. More precisely, if $0\leq \ancr\leq \frac{1}{2}\across$, then we can upper bound~\eqref{eq:apxWTAPcomb} by
\begin{equation}\label{eq:apxWTAPBound1}
1 + \min\left\{ 1 - \across, \ \across\left(1-\frac{\across}{48 p^2 \Delta^2}\right) \right\}\enspace.
\end{equation}
Similarly, if $\frac{1}{2}\across < \ancr \leq 1$, then~\eqref{eq:apxWTAPcomb} can be upper bounded by
\begin{equation}\label{eq:apxWTAPBound2}
1 + \min\left\{
   1 - \across\left(1+\frac{1}{2\Delta}\right), \ \across\right\}\enspace.
\end{equation}
Finally, the maximum value that~\eqref{eq:apxWTAPBound1} achieves over $\across\in [0,1]$ is given by
\begin{equation}\label{eq:apxWTAPex}
2 - 4 p\Delta \left(12 p\Delta - \sqrt{144 p^2\Delta^2 - 3}\right) \enspace,
\end{equation}
and one can show that this is strictly larger than the maximum value of~\eqref{eq:apxWTAPBound2} over $\across\in [0,1]$, for any values of $p,\Delta\geq 1$.
Hence,~\eqref{eq:apxWTAPex} is an upper bound on the approximation guarantee of the algorithm. Finally, the result 
follows by observing that the function $f(\beta)\coloneqq 2 - 4 \beta \left(12 \beta - \sqrt{144 \beta^{2} - 3}\right)$ 
is strictly increasing for $\beta\geq 1$ and satisfies $\lim_{\beta \rightarrow \infty} f(\beta) = 1.5$.
\end{proof}

\subsection*{Acknowledgments}
The authors are grateful to the anonymous reviewers for various helpful comments.

\bibliographystyle{plain}

\end{document}